\documentclass{aims}
\usepackage{amsmath}
\usepackage{amssymb}
  \usepackage{paralist}
  \usepackage{graphics} 
  \usepackage{epsfig} 
\usepackage{graphicx}  \usepackage{epstopdf}
\usepackage{caption}
 \usepackage[colorlinks=true]{hyperref}
\hypersetup{urlcolor=blue, citecolor=red}

\def\currentvolume{X}
 \def\currentissue{X}
  \def\currentyear{200X}
   \def\currentmonth{XX}
    \def\ppages{X--XX}
     \def\DOI{10.3934/xx.xx.xx.xx}

\newtheorem{theorem}{Theorem}[section]

\theoremstyle{definition}

\DeclareMathOperator*{\argmax}{\arg\!\max}

\graphicspath{{./Figures/}}

\title[Active Subspaces for a Long term Model of HIV]{Mathematical Analysis and Dynamic Active Subspaces for a Long term model of HIV}

\author[Tyson Loudon and Stephen Pankavich]{}


 \email{tloudon@mines.edu}
 \email{pankavic@mines.edu}

\thanks{The second author is supported by NSF grant DMS 12-11667}

\thanks{$^*$ Corresponding author: pankavic@mines.edu}

\begin{document}
\maketitle

\centerline{\scshape Tyson Loudon and Stephen Pankavich$^*$}
\medskip
{\footnotesize
 \centerline{Department of Applied Mathematics and Statistics}
  \centerline{Colorado School of Mines}
  \centerline{1500 Illinois St.}
  \centerline{Golden, CO 80401, USA}
} 

\bigskip

 \centerline{(Communicated by the associate editor name)}

\begin{abstract}
Recently, a long-term model of HIV infection dynamics \cite{longTerm} was developed to describe the entire time course of the disease.  It consists of a large system of ODEs with many parameters, and is expensive to simulate.  In the current paper, this model is analyzed by determining all infection-free steady states and studying the local stability properties of the unique biologically-relevant equilibrium. Active subspace methods are then used to perform a global sensitivity analysis and study the dependence of an infected individual's T-cell count on the parameter space. Building on these results, a global-in-time approximation of the T-cell count is created by constructing dynamic active subspaces and reduced order models are generated, thereby allowing for inexpensive computation.
\end{abstract}


\section{Introduction}

The Human Immunodeficiency Virus (HIV) disables many components of the body's immune system and, without antiretroviral treatment, leads to the onset of Acquired Immune Deficiency Syndrome (AIDS).
Despite the vast amount of study devoted to understanding viral pathogenesis and developing new therapeutics, no procedure or medication currently exists to reliably eliminate the virus from a host.  However, new advances in long-term treatment strategies and insight into disease dynamics have stemmed from mathematical and computational modeling approaches, in addition to clinical experimentation. 

A variety of mathematical models have been proposed to describe HIV infection and disease dynamics \cite{CP, K96, KP, KW, KWC, MKK, NowakMay, PS, P1, PN, TanWu}. Unfortunately, using a model to capture the entire time course of infection within the body can be troublesome as many oversimplify the biological dynamics of the disease in an effort to gain mathematical and rudimentary biological insight, and fail to capture all stages of infection. 
The majority of models accurately capture only the first stage(s) of infection \cite{CP, K96, KP, MKK, NowakMay, Pank, PP} with the T-cell count and viral load asymptotically approaching a nonzero limit - the latter referred to as the viral set point.
One recent description has been able to provide a holistic understanding of disease dynamics by accurately capturing all three stages of infection. This model, proposed in \cite{longTerm}, is comprised of a system of seven nonlinear autonomous differential equations that are fully coupled and augmented by twenty-seven distinct parameters.  
In this paper we investigate the dynamical properties of the model and establish a result concerning its large-time asymptotic behavior. Further, we analyze this model, utilizing mathematical and statistical methods to elucidate the contribution of the parameter space on an infected individual's T-cell count, and approximate solutions as a function of time.

The paper proceeds as follows.  Within the next section, all infection-free steady states of the model are determined and the local asymptotic stability properties of the biologically-relevant equilibrium are studied.  The associated theorems have interesting implications for the model's predictive nature, especially upon the introduction of antiretroviral therapy.  In Section $3$, we utilize active subspace methods to perform a global sensitivity analysis of the model with respect to its parameter space and further investigate the dependence of an infected individual's T-cell count on system parameters.  With this information, a global-in-time approximation of the T-cell count is constructed using dynamic active subspaces. Additionally, various reduced models are constructed to represent different stages of the disease, and these are discussed in detail.  In general, active subspace methods are a useful tool to perform global sensitivity analysis of a given parameter space, provide a clear picture of the most important activity in a model arising from parameter variation, construct dimensionally-reduced approximations to complex, dynamical models, and execute inexpensive numerical approximations from models that require computationally-intensive simulation. This is of particular importance in the current context as parameters for the original long-term HIV model are individual-dependent, and therefore must be determined for each new patient.  Hence, the construction of a significantly less expensive computational approximation will allow one to utilize the model for large sets of patient data.
Finally, we note that all of the MATLAB scripts and functions used to generate our results are provided free, open source, and available to the public at \url{http://inside.mines.edu/~pankavic/activeHIV}.

\section{Model Description, Parameters, and Analysis}
\label{ch:model}

To begin, we consider the following long term model of HIV disease dynamics within a host, as recently formulated in \cite{longTerm}: 
\begin{equation} 
\label{Hadji}
\left.
\begin{aligned}
\frac{dT}{dt} &= s_1 + \frac{p_1}{C_1+V} TV - \delta_1 T - (K_1 V + K_2 M_I)T \\
\frac{dT_I}{dt} &= \psi(K_1 V + K_2 M_I)T + \alpha_1 T_L - \delta_2 T_I - K_3 T_I CTL \\
\frac{dT_L}{dt} &= (1 - \psi)(K_1 V + K_2 M_I)T - \alpha_1 T_L - \delta_3 T_L \\
\frac{dM}{dt} &= s_2 + K_4 MV - K_5 MV - \delta_4 M \\
\frac{dM_I}{dt} &=  K_5 MV - \delta_5 M_I - K_6 M_I CTL \\
\frac{dCTL}{dt}  &=  s_3 + (K_7 T_I + K_8 M_I)CTL - \delta_6 CTL \\
\frac{dV}{dt} &= K_9 T_I + K_{10} M_I - K_{11} TV - (K_{12} + K_{13})MV - \delta_7 V 
\end{aligned}
\right \}
\end{equation}
The population of CD4$^+$ T-cells, denoted here by $T(t)$, is one of the most critical components in determining the body's response to HIV infection, and the first equation represents its time evolution. 
In (\ref{Hadji}) the T-cell population is increased by a standard source, $s_1$, which represents the constant supply rate of immunocompetent T-cells from the thymus, and a nonlinear generation term $\frac{p_1}{C_1+V}TV$, which accounts for the appearance of new T-cells due to the immune system's response to the infection. 
In contrast, because $T$-cells have a finite lifespan, a natural death term $\delta_1 T$ is also included.
The last two terms in the first equation model the infection of T-cells by either virions, the population of which is denoted by $V(t)$, or infected macrophages, $M_I(t)$, at rates $K_1$ and $K_2$, respectively.  The latter term is introduced in \eqref{Hadji} as studies have shown that infected macrophages likely play a vital role in the progression of infection \cite{Gumel, SIV} by producing large amounts of virus even after the T-cell population has been depleted.  

The second equation describes changes within the actively-infected T-cell population ($T_I$), the first term of which represents the addition to this compartment due to new infections.  However, not all interactions between virions and healthy T-cells produce actively-infected T-cells. Some proportion $\psi$ of these new infections contribute to the $T_I$ population, while the remainder ($1-\psi$) lack necessary host factors, resulting in the absence of viral protein expression, and become latently-infected T-cells ($T_L$), contributing to the third equation.
The infected T-cell population adds, in addition to the number created by the virions and macrophages, a supply of newly-activated latent T-cells, at a per-capita rate $\alpha_1$. In turn, this portion of the $T_L$ population is lost within the third equation.  Infected T-cells are lost due to natural death, which is represented by $\delta_2 T_I$, while a more interesting term, $K_3 T_I CTL$, represents the loss of infected T-cells due to cytotoxic lymphocytes ($CTL$), one of the many attacker cells the immune system employs.  Latently-infected T-cells also possess a natural death rate denoted by $\delta_3$. 

\begin{figure}[t]
	\vspace{-0.1in}
	\centering
	\includegraphics[scale=.55]{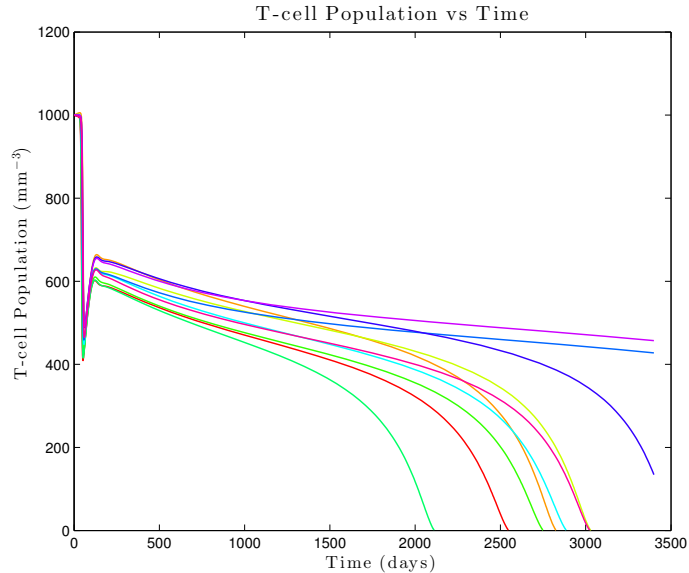}
	\vspace{-0.1in}
	\caption{Ten simulations of \eqref{Hadji} with representative parameter values.} 
	\label{sims}
	\vspace{-0.1in}
\end{figure}

Similar to T-cells, the macrophage population ($M$) possess a natural birth rate, $s_2$, as well as a natural death rate, $\delta_4$. 
As discussed in \cite{NowakMay}, macrophages divide and become more aggressive in seeking out pathogens once they are alerted of a viral presence by CD4+ T-cells. 
Hence, the macrophage population is increased in response to HIV infection at a rate of $K_4$. 
Macrophages also attempt to eliminate virions, but may also become infected, adding to the infected macrophage population ($M_I$). 
Infected macrophages naturally expire at a certain rate $\delta_5$, and are also destroyed by cytotoxic lymphocytes at a rate $K_6$. Once infected, macrophages produce virions at a rate given by $K_{10}$.
Infected macrophages may also infect healthy T-cells, and the rate at which this occurs is denoted by $K_2$.

The main defender of the body against infected cells is the cytotoxic lymphocyte, the population of which is denoted by $CTL$. These cells seek to destroy renegade T-cells and macrophages that have been infected and altered by HIV.  As for the other components of the immune system, we assume that cytotoxic lymphocytes are produced at a constant rate $s_3$ by the bone marrow. Additionally, new lymphocytes are produced proportionate to the other aspects of the body's immune response.  Hence, the term $(K_7T_I+K_8M_I)CTL$, occurs within the sixth equation to represent the recruitment of new lymphocytes. 

Lastly, the growth of the virion population $V$ depends on a variety of parameters. New virions are continually produced by infected T-cells and infected macrophages at rates $K_9$ and $K_{10}$, respectively. In addition, the virion population is decreased through a variety of means.  Due to infection, virions are lost at the rates $K_{11}$ and $K_{12}$, proportional to the interaction of T-cells and macrophages with virions, respectively. Also, macrophages, responding to antibodies produced by a host, ingest and destroy virions at a rate $K_{13}$. Finally, virus particles are cleared from the host by other factors, including the innate immune response, at a rate $\delta_7$. 

Note that all parameter values in \eqref{Hadji} are positive. Typical values and ranges for the parameters taken from \cite{longTerm} can be found within Table~\ref{tab:paramvalues}.
The time course of infection predicted by this model with the established parameter values was shown to agree well  with clinical data from \cite{Fauci, Greenough, Pennisi}, and representative simulations of \eqref{Hadji} with realistic parameter values are displayed in Figure \ref{sims}.

\begin{table}[t]
\centering
\begin{tabular}{| c | c | c | c | c | c |}
\hline
Parameter & Value & Range & Value taken from: & Units \\
\hline \hline
$s_1$ & 10 & 5 - 36 & \cite{KP} & mm$^{-3}$d$^{-1}$ \\
\hline
$s_2$ & 0.15 & 0.03 - 0.15 & \cite{KP} & mm$^{-3}$d$^{-1}$ \\
\hline
$s_3$ & 5 & - & \cite{longTerm} & mm$^{-3}$d$^{-1}$ \\
\hline
$p_1$ & 0.2 & 0.01 - 0.5 & \cite{longTerm} & d$^{-1}$ \\
\hline
$C_1$ & 55.6 & 1 - 188 & \cite{longTerm} & mm$^{-3}$ \\
\hline
$K_1$ & 3.87 x $10^{-3}$ & 10$^{-8}$ - 10$^{-2}$ & \cite{longTerm} & mm$^{3}$d$^{-1}$ \\
\hline
$K_2$ & $10^{-6}$ & $10^{-6}$ & \cite{KP} & mm$^{3}$d$^{-1}$ \\
\hline
$K_3$ & 4.5 x 10$^{-4}$ & 10$^{-4}$ - 1 & \cite{longTerm} & mm$^{3}$d$^{-1}$ \\
\hline
$K_4$ & 7.45 x 10$^{-4}$ & - & \cite{longTerm} & mm$^{3}$d$^{-1}$ \\
\hline
$K_5$ & 5.22 x 10$^{-4}$ & 4.7 x 10$^{-9}$ - 10$^{-3}$ & \cite{longTerm} & mm$^{3}$d$^{-1}$ \\
\hline 
$K_6$ & 3 x 10$^{-6}$ & - & \cite{longTerm} & mm$^{3}$d$^{-1}$ \\
\hline
$K_7$ & 3.3 x 10$^{-4}$ & 10$^{-6}$ - 10$^{-3}$ & \cite{longTerm} & mm$^{3}$d$^{-1}$ \\
\hline
$K_8$ & 6 x 10$^{-9}$ & - & \cite{longTerm} & mm$^{3}$d$^{-1}$ \\
\hline
$K_9$ & 0.537 & 0.24 - 500 & \cite{longTerm} & d$^{-1}$ \\
\hline
$K_{10}$ & 0.285 & 0.005 - 300 & \cite{longTerm} & d$^{-1}$ \\
\hline
$K_{11}$ & 7.79 x 10$^{-6}$ & - & \cite{longTerm} & mm$^{3}$d$^{-1}$ \\
\hline
$K_{12}$ & 10$^{-6}$ & - & \cite{longTerm} & mm$^{3}$d$^{-1}$ \\
\hline
$K_{13}$ & 4 x 10$^{-5}$ & - & \cite{longTerm} & mm$^{3}$d$^{-1}$ \\
\hline
$\delta_1$ & 0.01 & 0.01 - 0.02 & \cite{longTerm} & d$^{-1}$ \\
\hline
$\delta_2$ & 0.28 & 0.24 - 0.7 & \cite{longTerm} & d$^{-1}$ \\
\hline 
$\delta_3$ & 0.05 & 0.02 - 0.069 & \cite{longTerm} & d$^{-1}$ \\
\hline
$\delta_4$ & 0.005 & 0.005 & \cite{KP} & d$^{-1}$ \\
\hline
$\delta_5$ & 0.005 & 0.005 & \cite{KP} & d$^{-1}$ \\
\hline
$\delta_6$ & 0.015 & 0.015 - 0.05 & \cite{params2} & d$^{-1}$ \\
\hline
$\delta_7$ & 2.39 & 2.39 - 13 & \cite{KP} & d$^{-1}$ \\
\hline
$\alpha_1$ & 3 x 10$^{-4}$ & - & \cite{longTerm} & d$^{-1}$ \\ 
\hline
$\psi$ & 0.97 & 0.93 - 0.98 & \cite{longTerm} & - \\
\hline 
\end{tabular} 
\caption{Parameter values and ranges}
\label{tab:paramvalues}
\vspace{-0.2in}
\end{table}

\subsection{Infection-free steady states}

Though the system (\ref{Hadji}) possesses a large number of steady states - the authors have discovered at least ten using standard parameter values and a computational root finder - one is often most interested in understanding the dynamical properties of the disease-free equilibrium. In this section, we identify such equilibria and investigate the stability of the biologically-relevant state.  
Our first result demonstrates that only one such equilibrium state exists when all populations of \eqref{Hadji} are positive.
\begin{theorem}
\label{T1}
The model \eqref{Hadji} possesses exactly two virus-free (i.e. $V \equiv 0$) steady states.  One of these states, namely
$$E := \left ( \frac{s_1}{\delta_1 - \omega K_2 K_9} , \omega K_{10} , \frac{\omega K_{10} \xi}{K_6 (\alpha_1 + \delta_3 \psi)} , \frac{s_2}{\delta_4} , -\omega K_9 , -\frac{\delta_5}{K_6} , 0 \right ) $$
achieves negative values, where $$ \omega = \frac{s_3 K_6 + \delta_5 \delta_6}{\delta_5 (K_7 K_{10} - K_8 K_9)} \quad and \quad \xi = (1 - \psi)(\delta_2 K_6 - \delta_5 K_3).$$
The only nonnegative (i.e. biologically relevant) steady state of (\ref{Hadji}) satisfying $V \equiv 0$ is 
$$ E_{NI} := \left (\frac{s_1}{\delta_1}, 0, 0, \frac{s_2}{\delta_4}, 0, \frac{s_3}{\delta_6}, 0 \right ).$$
\end{theorem}
Hence, the only guarantee of viral clearance as $t \to \infty$ occurs when actively and latently infected populations are also eradicated, resulting in healthy T-cell and macrophage populations tending asymptotically to background values.
The proof of Theorem~\ref{T1} is contained in Appendix A. Utilizing standard parameter values for this model from Table~\ref{tab:paramvalues}, we find the following equilibrium populations for $E$: 
$$
\begin{gathered}
T = 1010.39 \, {\rm mm}^{-3}, \qquad
T_I = 54.57 \, {\rm mm}^{-3}  \\
T_L = -15.77 \, {\rm mm}^{-3}, \qquad
M = 30 \, {\rm mm}^{-3} \\
M_I = -102.83 \, {\rm mm}^{-3}, \quad
CTL = -1666.67 \, {\rm mm}^{-3}, \quad V = 0 \, \rm{mm}^{-3} \\
\end{gathered}
$$
Since the parameter values in \eqref{Hadji} are positive, the steady state $E$ given in Theorem~\ref{T1} must have a negative cytotoxic T-lymphocyte population, namely $-\frac{\delta_5}{K_6}$. So, under no parameter regime will $E$ be biologically relevant.  With the unique infection-free steady state identified, we turn to its stability properties.

\subsection{Stability Analysis}

Next, we provide necessary and sufficient conditions which guarantee the local asymptotic stability of the disease-free equilibrium $E_{NI}$.\\
\begin{theorem}
\label{T3}
The equilibrium state $E_{NI} $ is locally asymptotically stable if and only if $R_0 \leq 1$, where
$$R_0 = \max \{ R_1, R_2, R_3 \}$$
\rm{and}
$$R_1 = \frac{K_1K_9}{\delta_2K_{11}}, \qquad R_2 = \frac{K_5K_{10}}{(K_{12} + K_{13} ) \delta_5}$$
$$R_3 = \frac{K_2 K_5 K_9 s_1 s_2}{\delta_1 \delta_2 \delta_4 \delta_5 \delta_7 + \delta_4 \delta_5 K_1 K_9 s_1 + \delta_1 \delta_2 K_5 K_{10} s_2}.$$
\end{theorem}

The proof of Theorem~\ref{T3} is also contained within Appendix A. Computing the basic reproduction number of Theorem \ref{T3} by using the standard parameter values given in Table~\ref{tab:paramvalues}, we find that $R_0 = R_1 = 953 >> 1$. Hence, as expected, the non-infective steady state $E_{NI}$ is not locally asymptotically stable.  

Introducing antiretroviral therapy, or ART, into the system provides additional insight into this result.
Two specific classes of ART drugs, namely Reverse Transcriptase Inhibitors (RTIs) and Protease Inhibitors (PIs), serve to reduce the amount of new virus produced by either reducing the ability of virions to replicate through reverse transcription or disabling the capability of newly-produced virions to mature, thereby rendering them uninfective.  The efficacies of these classes of drug, denoted $\epsilon_{RTI}, \epsilon_{PI} \in [0,1]$, can be incorporated using the transformations 
$$K_1 \to K_1 (1 - \epsilon_{RTI}), \ \ K_5 \to K_5 (1 - \epsilon_{RTI}), \ \ K_9 \to K_9 (1 - \epsilon_{PI}), \ \ K_{10} \to K_{10} (1 - \epsilon_{PI}).$$ 
As these constants appear within the stability result and decrease each of the ratios $R_1$, $R_2$, and $R_3$, it follows that a large enough efficacy will force the system to tend towards $E_{NI}$ as $t \to \infty$.  Therefore, one merely needs to perform this transformation on the result of Theorem \ref{T3} in order to determine what efficacy is needed to guarantee viral clearance.  Using the previously-determined parameter values and defining $\epsilon \in [0,1]$ by
$$1-\epsilon :=  (1-\epsilon_{RTI}) (1 - \epsilon_{PI}),$$ 
we find $\epsilon > 0.998$ in order to force $R_0 < 1$.  Thus, a combined drug efficacy greater than $99.8\%$ would be needed to asymptotically drive the system to clearance.  Of course, lesser drug efficacies could still give rise to viral loads that are effectively negligible rather than tending to zero, and therefore correspond to viral clearance.

Notice that the asymptotic stability result in Theorem~\ref{T3} depends only upon a relatively few number ($15$ of $27$) of the parameters.  Therefore, the majority of the pertinent dynamics takes place on a lower-dimensional subspace of the entire parameter space.  Hence, in the next section we explore a dynamic tool to better understand the contribution of the parameter space to populations within the model. In particular, this will lead to reducing the dimension of the parameter space with minimal loss of information using an active subspace decomposition.

\section{Dynamic Active Subspaces, Sensitivity, and Reduced Models}
\label{ch:dynamicaspace}
In this section we will use active subspace methods to approximate the T-cell count at a specific time given the parameter values in \eqref{Hadji}. We closely follow the material as developed by Constantine and co-authors \cite{Aspaces, compAspaces}.  An active subspace is a low-dimensional linear subspace of the set of parameters, in which input perturbations along these directions alter the model's predictions more, on average, than perturbations which are orthogonal to the subspace. These subspaces allow for a global measurement of sensitivity of output variables with respect to parameters, and often the construction of reduced-order models that greatly decrease the dimension of the parameter space.

\subsection{Active Subspace Methods}
\label{subsec:aspace}
The general structure of an active subspace decomposition begins by letting $m \in \mathbb{N}$ be given and defining the space $X = [-1,1]^m$.  Also given is a differentiable function $f: X \to \mathbb{R}$ and an associated probability density $\rho: X \to \mathbb{R}^+$ satisfying $$\int \rho(x) \ dx = 1.$$
Here, the space $X$ represents a normalized set of parameter values.  With these quantities in place, consider the matrix $C$ defined by
\begin{equation}
\label{C}
C = \int (\nabla_x f) (\nabla_x f)^T \! \rho \, d\bf{x}.
\end{equation}
For any smooth $f$, the matrix $C$ represents an average derivative functional which weights input values according to the density $\rho$.  In general terms, $f(x)$ represents the quantity of interest in a given model, while gradients of $f$ are taken with respect to normalized model parameters $x \in X$, and $\rho(x)$ is the probability density associated to the values of these parameters. 
We note here that a single normalized parameter is a random variable taking values in $[-1,1]$, which when appropriately scaled represents a parameter in the original model \eqref{Hadji}. 
Since the dimension of the parameter space in this model is $27$, we take $m = 27$ throughout.  The matrix $C$ is the average of the outer product of the gradient of $f$ with itself and has some useful properties that will allow us to deduce information about how $f$ is altered by perturbations in its arguments. 

Considering each entry of the matrix  
\begin{equation*}
C_{ij} = \int \frac{\partial f}{\partial x_i} \frac{\partial f}{\partial x_j} \rho(\mathbf{x}) \, d \mathbf{x}
\end{equation*}
we note that $C$ is symmetric, and thus permits the spectral eigendecompostion 
\begin{equation}
\label{edecomp}
C = W \Lambda W^T, \quad \mathrm{where} \quad \Lambda = \mathrm{diag}(\lambda_1,\ldots,\lambda_m), \quad \lambda_1 \geq \ldots \geq \lambda_m \geq 0.
\end{equation}
Here, $W$ is an orthogonal matrix whose columns $\mathbf{w}_i, \, (i = 1, \ldots, m)$ are the orthonormal eigenvectors of $C$. 
From \eqref{edecomp} we can further solve for the eigenvalues of $C$, which are given by
\begin{equation}
\label{evalues}
\lambda_i = \int \big((\nabla_x f)^T \textbf{w}_ i\big)^2 \rho(\mathbf{x})  \, d\textbf{x}, \quad i = 1, \ldots, m.
\end{equation}
From \eqref{evalues} we see that the eigenvalues of the $C$ matrix are the mean squared directional derivatives of $f$, in the direction of the corresponding eigenvector. Thus, the eigenvalues of $C$ provide useful information about the quantity of interest. For instance, if a particular eigenvalue is small then \eqref{evalues} tells us that, on average, $f$ does not change significantly in the direction of the corresponding eigenvector. Conversely, if the eigenvalue under consideration is large, then we may deduce that $f$ changes considerably in the direction of the corresponding eigenvector. Therefore, it will be of interest to further investigate the behavior of the function in this direction. 

Once the eigendecomposition \eqref{edecomp} has been determined, the eigenvalues and eigenvectors can be  separated in the following way:
\begin{equation}
\label{mdecomp}
\Lambda = \begin{bmatrix} 
			\Lambda_1 & 0 \\
			0 & \Lambda_2 \\
		  \end{bmatrix} , \quad
W = \begin{bmatrix}
         W_1 & W_2 \\
	\end{bmatrix}.
\end{equation}
where $\Lambda_1$ contains the ``large'' eigenvalues of $C$, $\Lambda_2$ contains the ``small'' eigenvalues, and $W_k$ contains the eigenvectors associated with each $\Lambda_k$, for $k=1, 2$. 
An easy way to differentiate between the ``large'' and ``small'' eigenvalues is to list them on a log plot from greatest to least and determine a spectral gap. This gap will correspond to differences of at least an order of magnitude, and thus allow one to compartmentalize large eigenvalues within $\Lambda_1$ and the remaining smaller eigenvalues in $\Lambda_2$. A more systematic method of choosing how many eigenvalues to store within $\Lambda_1$ will be presented in Section~\ref{subsec:T1700}.

With the decomposition \eqref{mdecomp}, we can represent any element $\mathbf{x}$ of the parameter space by
\begin{equation}
\mathbf{x} = \underbrace{W W^T}_\text{I} \mathbf{x} = W_1 \underbrace{W_1^T \bf{x}}_\text{\bf{y}} + W_2 \underbrace{W_2^T \bf{x}}_\text{\bf{z}} = W_1 \mathbf{y} + W_2 \mathbf{z}.
\end{equation}
Thus, evaluating the quantity of interest at $\mathbf{x}$ is equivalent to doing so at the point $W_1 \mathbf{y} + W_2 \mathbf{z}$, i.e.
\begin{equation*}
f(\mathbf{x}) = f(W_1 \mathbf{y} + W_2 \mathbf{z}).
\end{equation*}
By the definition of $W_1$ and $W_2$ it's clear that small perturbations in $\mathbf{z}$ will not, on average, alter the values of $f$. However, small perturbations in $\mathbf{y}$ will, on average, change $f$ significantly. For this reason we define the range of $W_1$ to be the \textit{active subspace} of the model and the range of $W_2$ to be the corresponding \textit{inactive subspace}.  The linear combinations that generate these subspaces will then represent the contributions of differing parameters in the model and describe the sensitivity of the quantity of interest with respect to parameter variations.

\begin{figure}[t]
	\vspace{-0.3in}
\centering
	\includegraphics[scale = .4]{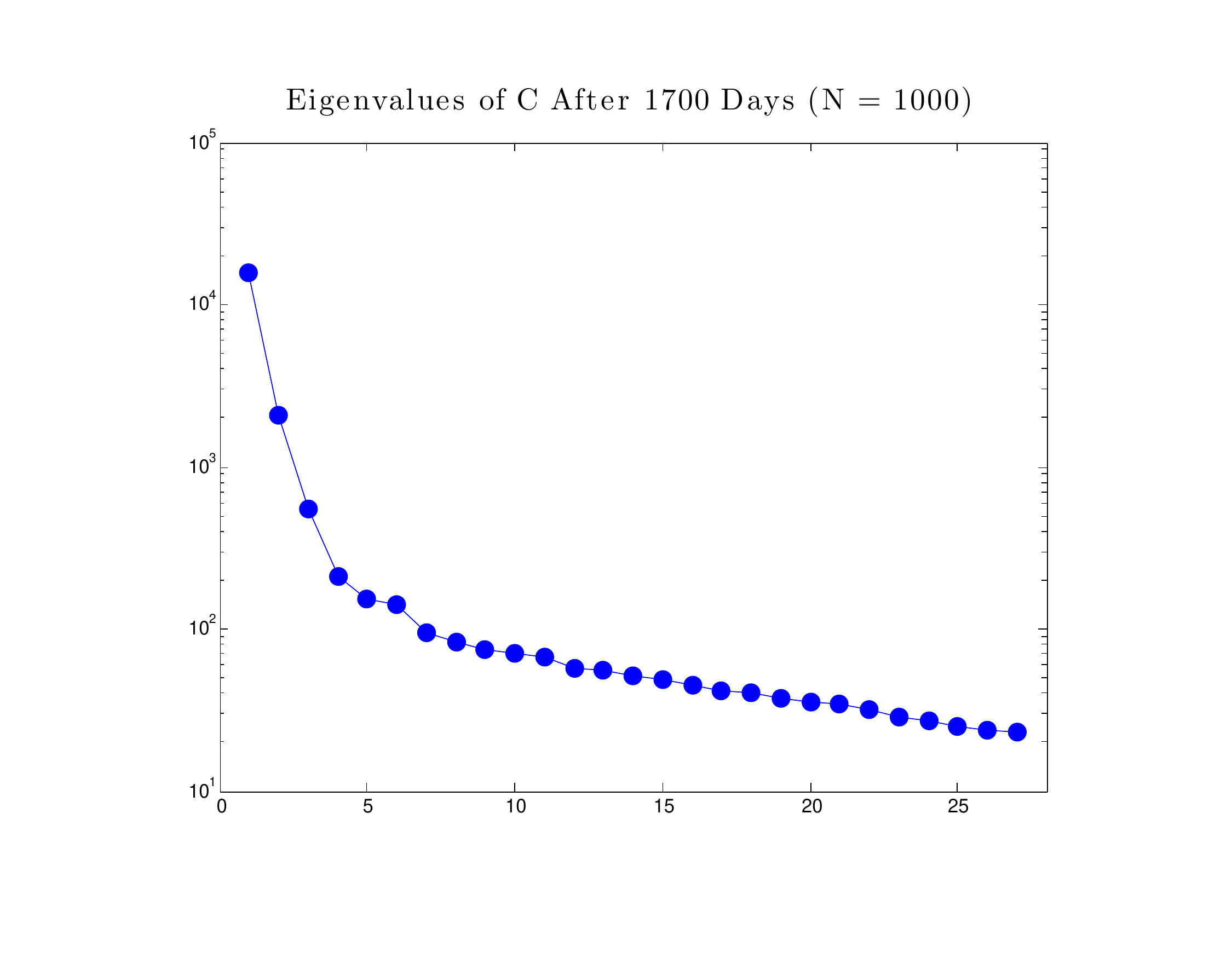}
	\vspace{-0.5in}
	\caption{Approximation of eigenvalues of $C$ using 1000 random samples.} 
	\label{fig:evalues1700}
\end{figure}

In general, the eigenvalues and eigenvectors of $C$ defined by \eqref{C} can be well-approximated, using a random sampling algorithm. We will briefly outline the method, but full details can be found in \cite{Aspaces} (Algorithm 3.1) and \cite{compAspaces}.  The algorithm can be described concisely as follows:
%
\begin{enumerate} 
\item Draw $N$ samples $\{ \mathbf{x}_j \}_{j=1}^N$ independently according to the density $\rho$. 
\item For each parameter sample $\mathbf{x}_j$, approximate the gradient $\nabla_x f_j = \nabla_x f(\mathbf{x}_j)$ using the finite difference, i.e. $$\partial_{x_i} f(\mathbf{x}_j) \approx \frac{f(\mathbf{x}_j + \mathbf{h}_i) - f(\mathbf{x}_j)}{\vert \mathbf{h}_i \vert}$$ where 
$$(\mathbf{h}_i)_k = 
\begin{cases}
\delta & \mbox{if } i = k \\ 
0 & \mbox{if } i \neq k .
\end{cases}$$ 
represents a vector perturbation from the sampled parameter values and $\delta>0$ can be taken arbitrarily small.
\item Approximate the matrix $C$ by 
\vspace{-20pt} 
\begin{center}
\begin{equation*}
C \approx \hat{C} = \frac{1}{N} \sum_{j = 1}^N  (\nabla_x f_j) (\nabla_x f_j)^T
\end{equation*}
\end{center}
\item Compute the eigendecompositions $\hat{C} = \hat{W}\hat{\Lambda}\hat{W}^T$.
\end{enumerate}

We note that the last step is equivalent to computing the singular value decomposition of the matrix
\begin{equation}
\frac{1}{\sqrt{N}}[ \nabla_x f_1 \ldots \nabla_x f_N] = \hat{W} \sqrt{\hat{\Lambda}} \hat{V},
\end{equation}
where it can be shown that the singular values are the square roots of the eigenvalues of $\hat{C}$ and the left singular vectors are the eigenvectors of $\hat{C}$. The singular value decomposition method of approximating $\hat{C}$ was developed first in \cite{Russi}. 

\subsection{An Illustrative Example - Approximating $\mathbf{T(1700)}$}
\label{subsec:T1700}
Next, we will demonstrate the active subspace method at one specific point in time by applying the aforementioned algorithm to the HIV model \eqref{Hadji}.  A suitable quantity of interest would likely involve the T-cell or virus population evaluated at a fixed time.  For this example we select $T(1700)$, the T-cell count $1700$ days after initial infection. This quantity was chosen because the T-cell count is a strong indicator of a patient's overall health and the most important factor in the progression of HIV infection. The time of $1700$ days after initial infection was chosen because regardless of parameter values, preliminary simulations have shown that the patient's T-cell count within \eqref{Hadji} will not have decreased to zero by this time. However, if a later time is chosen, the patient's T-cell count could vanish before the final time is reached. 

\begin{figure}[t]
	\vspace{-0.3in}
	\includegraphics[scale=.4]{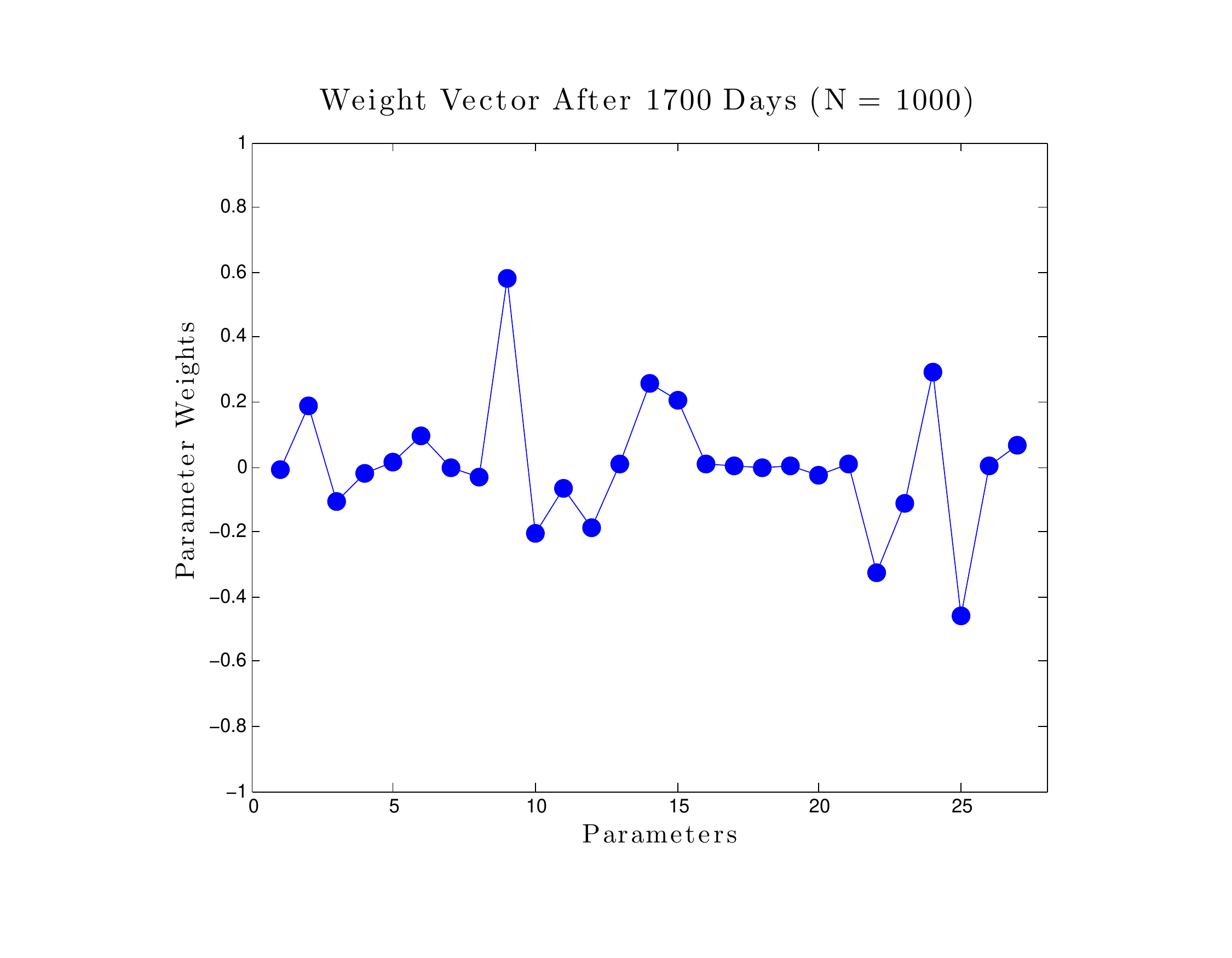}
	\vspace{-0.3in}
	\caption{Approximation of the 1st eigenvector of $C$ using $1000$ random samples. This is referred to as the first \textit{active variable vector} and denoted by $\mathbf{w}$.} 
	\label{fig:wv1700}
\end{figure}
In order to compute $C$, one must be able to construct the gradient of the quantity of interest with respect to the normalized parameter space as required in Step 3 of the algorithm. 
Since we do not have an explicit representation for $T(1700)$ as a function of system parameters, we cannot explicitly compute gradients. Instead we approximate them using the aforementioned finite difference scheme with a step size of $\delta = 10^{-6}$.   Additionally, each sample is chosen so that every normalized parameter is uniformly distributed between $-1$ and $1$, i.e. $\mathbf{x_j} \sim (U[-1,1])^{27}$. 
In order to map this normalized parameter space onto the biologically-relevant range of parameter values, we use the linear mapping
\begin{equation}
\label{p}
\mathbf{p} = \frac{1}{2}\big (\rm{diag}(\mathbf{x_u} - \mathbf{x_l})\mathbf{x_j} + (\mathbf{x_u} + \mathbf{x_l}) \big ),
\end{equation}
for each of the random samples $\mathbf{x_j} \sim (U[-1,1])^{27}$,
where $\mathbf{x_u}$ and $\mathbf{x_l}$ are vectors containing the upper and lower bounds on the parameters, respectively. Thus, the resulting vector $\mathbf{p}$ represents the actual parameter values input within the model. Next, a stiff differential equations solver (MATLAB's \texttt{ode23s} function) is used to compute the T-cell count after $1700$ days. Then, each of the $27$ parameters is perturbed by $\delta = 10^{-6}$, and again the corresponding T-cell count after $1700$ days is computed. With these two values, the aforementioned finite difference approximation is used to calculate the gradient of $f(\textbf{x}) = T(1700; \textbf{x})$ with respect to the normalized model parameters.
Regarding the upper and lower limits, $\mathbf{x_u}$ and $\mathbf{x_l}$ are taken to be $2.5\%$ above and below the standard values given in Table~\ref{tab:paramvalues}. 
%
Figure \ref{fig:evalues1700} displays the approximate eigenvalues of the corresponding $C$ matrix. Clearly a spectral gap exists between the first and second eigenvalues. 
\begin{figure}[t]
	\vspace{-0.3in}
\centering
	\includegraphics[scale=.40]{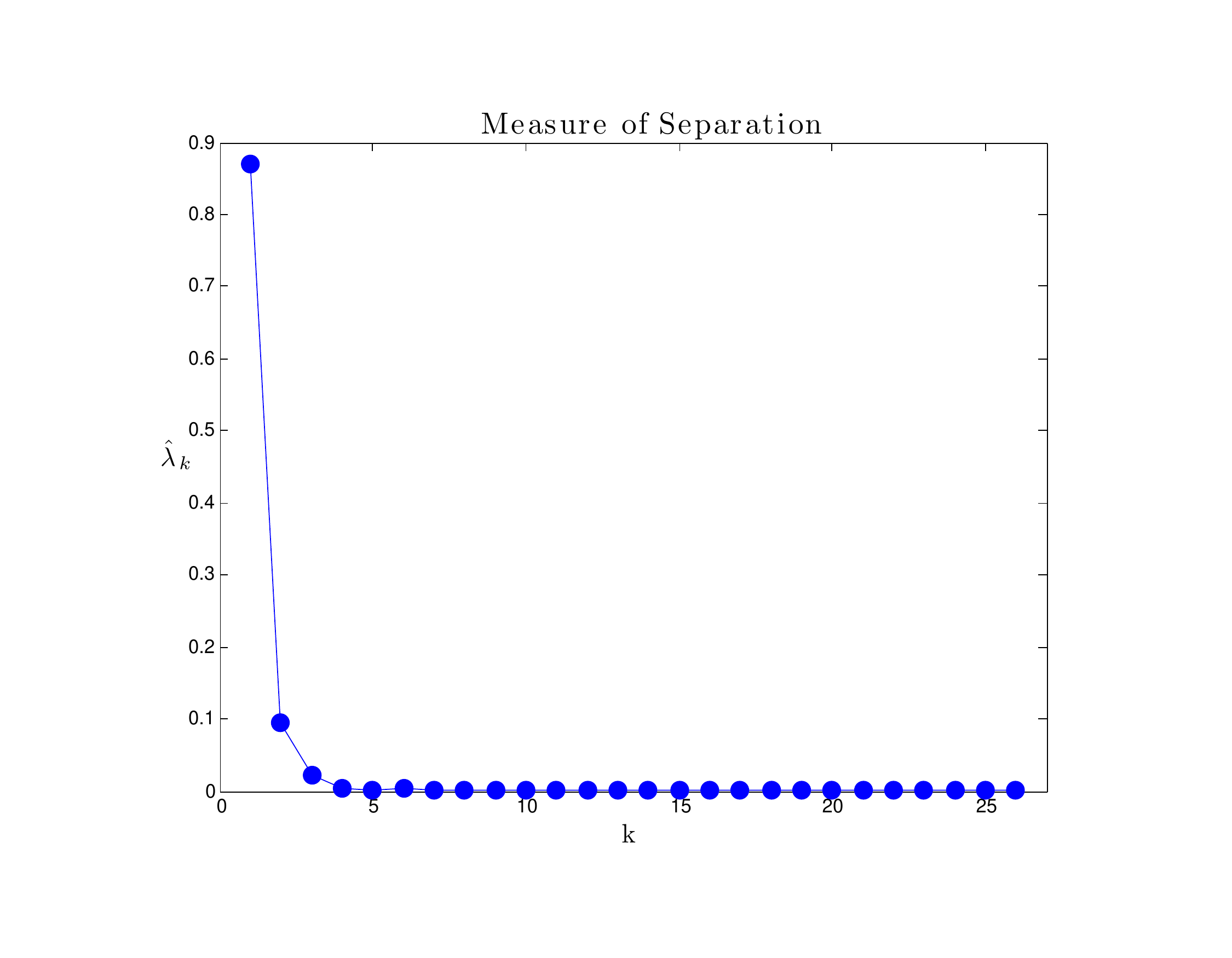}
	\vspace{-0.5in}
	\caption{Measure of separation for the eigenvalues of C.} 
	\label{fig:T1700sep}
\end{figure}
In order to automatically determine the optimal decomposition of $\Lambda$ we can use the relative measure of separation given by
\begin{equation}
\label{eq:sep}
\hat{\lambda}_k = \frac{\lambda_k - \lambda_{k+1}}{\lambda_1}, \quad k = 1,2, \ldots,26.
\end{equation}
Then, the dimension of the active subspace, i.e. the number of eigenvalues stored within $\Lambda_1$, will be given by
\begin{equation}
\label{dim}
\rm{dim} = \argmax_{k = 1, \ldots, 26} \hat{\lambda}_k.
\end{equation}
While the index of the largest value of $\hat{\lambda}_k$ describes the location of the largest spectral gap, it is often convenient to consider only the first two values $\hat{\lambda}_1$ and $\hat{\lambda}_2$. Doing so limits the dimension of the active subspace to one and two respectively, which allows for easy visualization of the quantity of interest as a function of the active subspace and allows one to fit a curve or surface to the data. Plotting the values of $\hat{\lambda}_k$ results in Figure~\ref{fig:T1700sep}.

Clearly, with this measure of separation, the optimal choice for the dimension of the active subspace is merely one. Consequently, we store $\lambda_1$ in the matrix $\Lambda_1$ and the remaining eigenvalues $\lambda_i$, $i = 2,\ldots,27$, along the diagonal of $\Lambda_2$. The active subspace is then generated by linear combinations of the entries of $\mathbf{w}$, the first eigenvector.
Figure~\ref{fig:wv1700} displays the eigenvector corresponding to the maximal eigenvalue shown in Figure~\ref{fig:evalues1700}, and it can be seen that three parameters possess associated weights greater than $0.3$. These are the $9$th, $22$nd, and $25$th parameters as ordered in Table~\ref{tab:paramvalues}.  From the table these parameters can be identified as $K_4$, $\delta_4$, and $\delta_7$, which represent the increase in macrophage population due to the immune system, the death rate of the macrophage population, and the death rate of the virus population, respectively. Hence, small perturbations in these parameters will significantly alter the value of $T(1700)$, as they are the most heavily weighted. Contrastingly, changes within the remaining parameters, whose weights are near zero, will not have an appreciable affect on $T(1700)$.

Plotting $T(1700)$ along the active subspace results in Figure~\ref{fig:SSP1700}. Here, the horizontal axis is represented by values of the first active variable $y = \mathbf{x} \cdot \mathbf{w}$, which represents a linear combination of the normalized parameters $\mathbf{x}$ with weights given by entries of the first active variable vector $\mathbf{w}$. Adopting the terminology of \cite{Aspaces}, we will refer to plots of the quantity of interest along the active subspace as \emph{sufficient summary plots}. Figure~\ref{fig:SSP1700} (left) shows a clear trend, namely that $T(1700;y)$ is a decreasing function of the active variable. 
\begin{figure}[t]
\vspace{-0.3in}
\hspace{-46mm}
\centering
\begin{minipage}{0.42\textwidth}
\centering
\includegraphics[scale = 0.42]{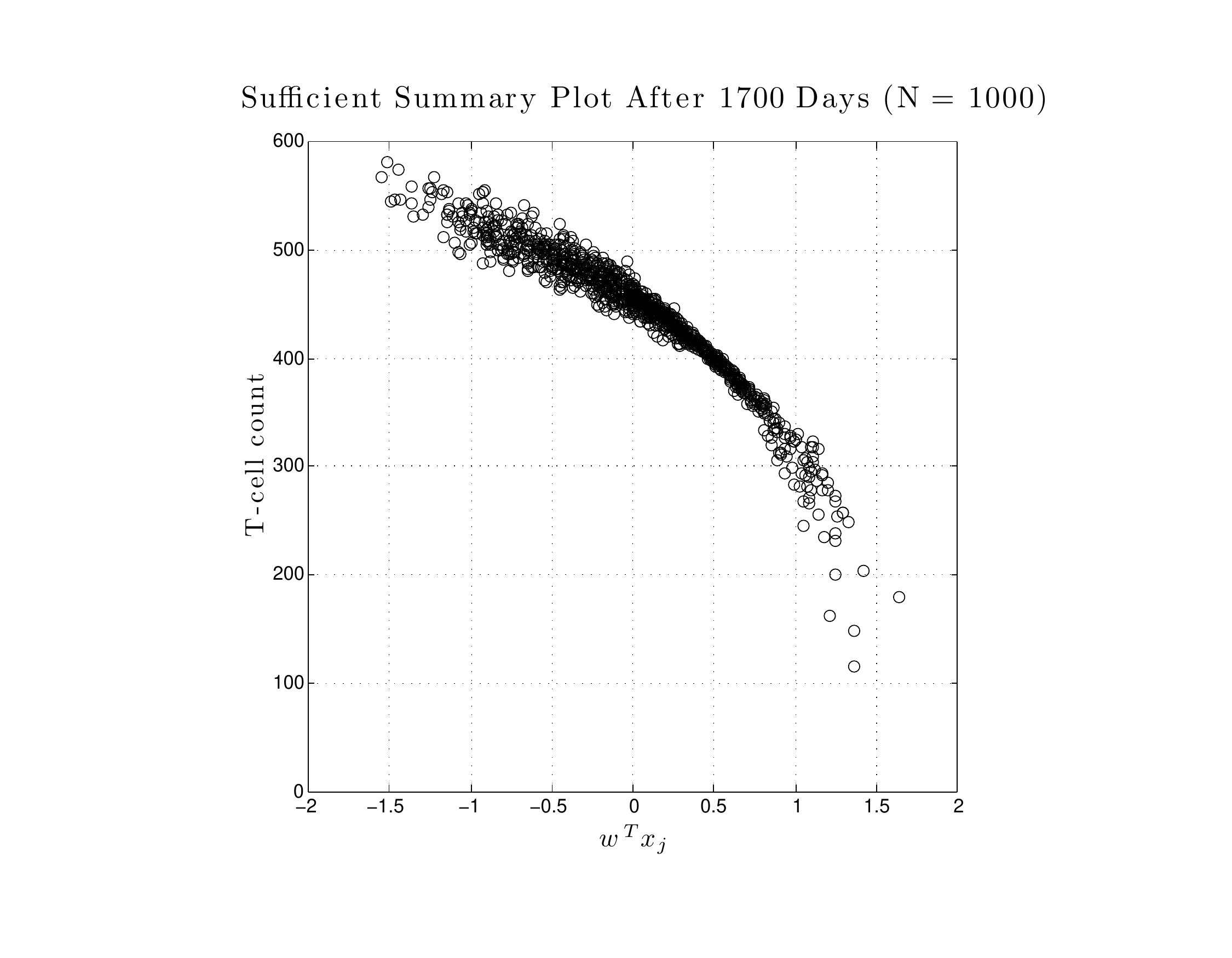}
\end{minipage}
\hspace{15mm}
\begin{minipage}{0.42\textwidth}
\centering
\hspace{-10mm}
\includegraphics[scale = 0.42]{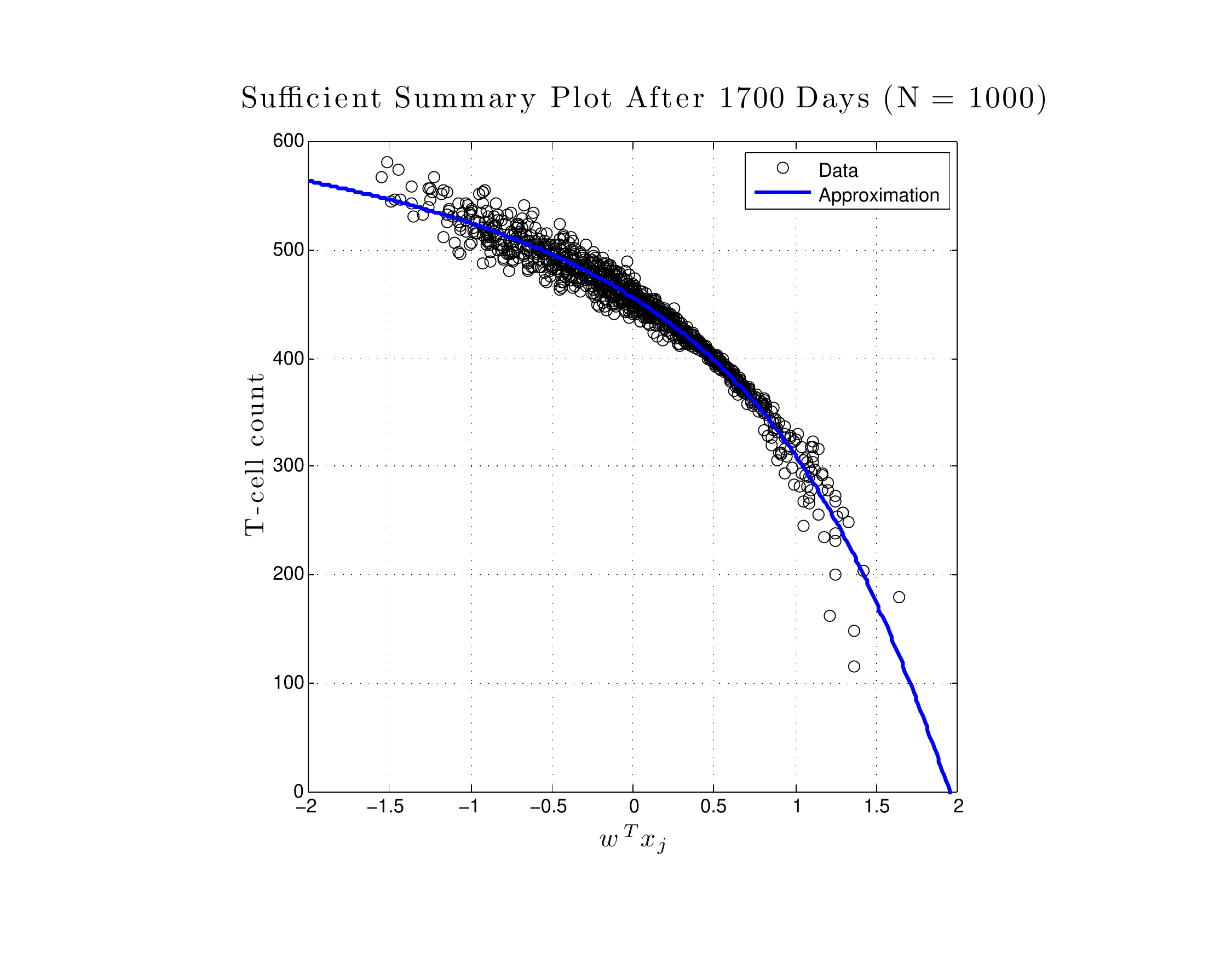}
\end{minipage}
\vspace{-0.4in}
\caption{Sufficient summary plot after 1700 days (left). Approximation to the T-cell count after 1700 days (right).}
\label{fig:SSP1700}
\end{figure}
%
%
Additionally, the data points within the sufficient summary plot can be fit to a particular function. In this case, a four parameter arctangent function is fit to the data to approximate the quantity of interest. This was performed using the MATLAB function \texttt{lsqcurvefit} which minimizes the residual (in the least squares sense) of the difference in the data and the approximation. Hence, we have determined that the data is best fit by
\begin{equation}
\label{analytic}
T(1700; y) = -79.2532 - 492.5680 \ {\rm tan}^{-1}(0.8933y  - 1.9069),
\end{equation}
and the resulting curve can be found in the plot on the right side of Figure~\ref{fig:SSP1700}.
%
%
In order to test the accuracy of the nonlinear approximation, $100$ simulations were run and the relative error computed. The results are displayed in Figure~\ref{fig:T1700errors} and indicate that the approximation is typically within $5\%$ from the value computed by computationally solving \eqref{Hadji} to determine $T(1700)$. Stated another way, over $90\%$ of the test simulations returned approximations whose error was $5\%$ or less. 

\begin{figure}[t]
	\vspace{-0.2in}
	\includegraphics[width = 0.6\textwidth]{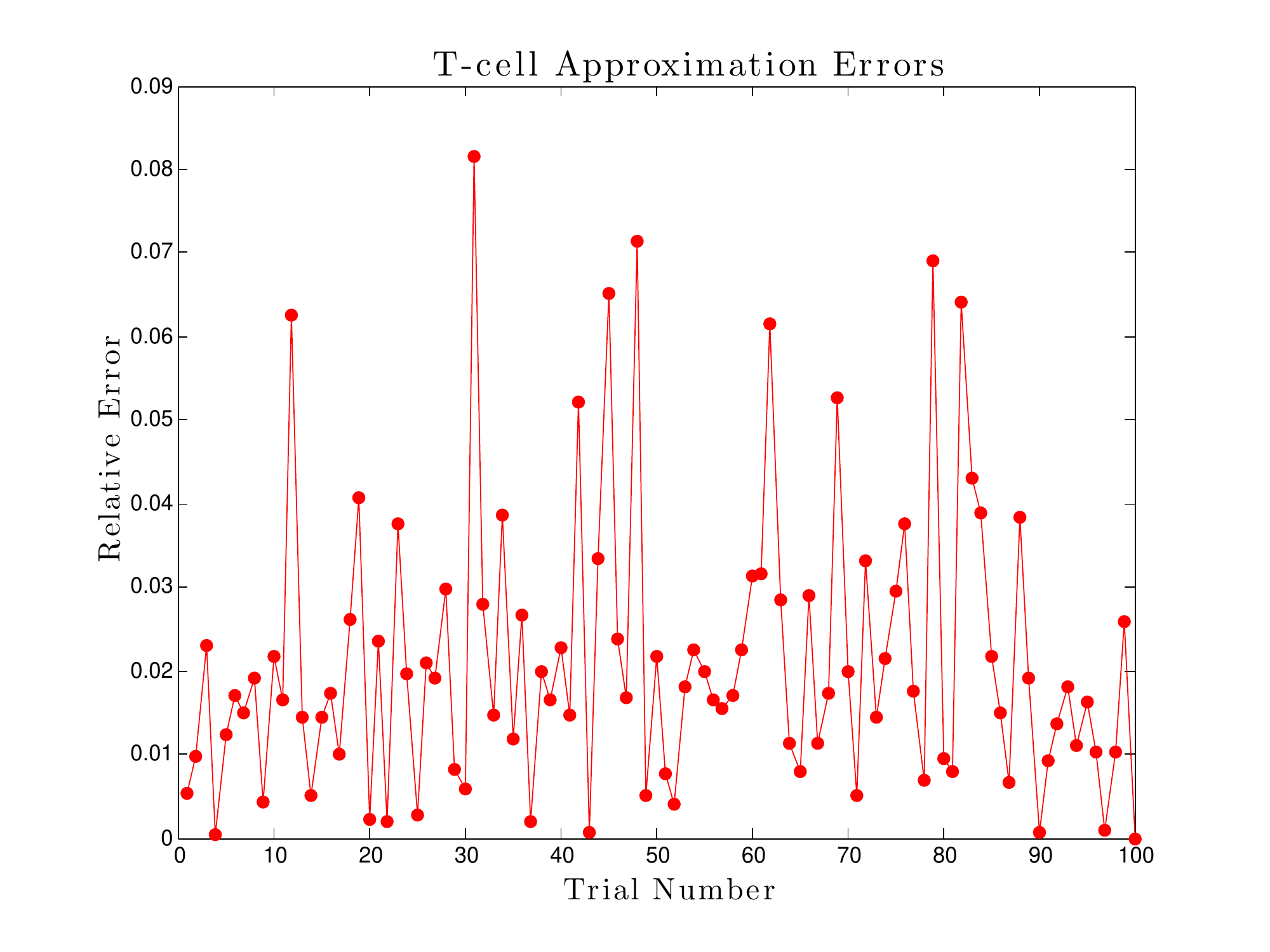}
	\vspace{-0.25in}
	\caption{Relative errors in the approximation of $T(1700)$.} 
	\label{fig:T1700errors}
	\vspace{-0.1in}
\end{figure}

Of course, not every eigenvalue decomposition will necessarily result in a one-dimensional active subspace of the model parameters.
For instance, should one wish to approximate the T-cell count $2000$ days after initial infection, the same order reduction method can be employed and the eigenvalues of the $C$ matrix can be computed.
The plot of the eigenvalues of $C$ can be seen in Figure~\ref{fig:GlobalSep} (left).
In this case the largest spectral gap exists between the second and third eigenvalues rather than the first and second, and therefore a more descriptive choice for the dimension of the active subspace is two instead of one.


\begin{figure}[t]
\hspace{-20mm}
\centering
\begin{minipage}{0.40\textwidth}
\centering
\includegraphics[scale = .40]{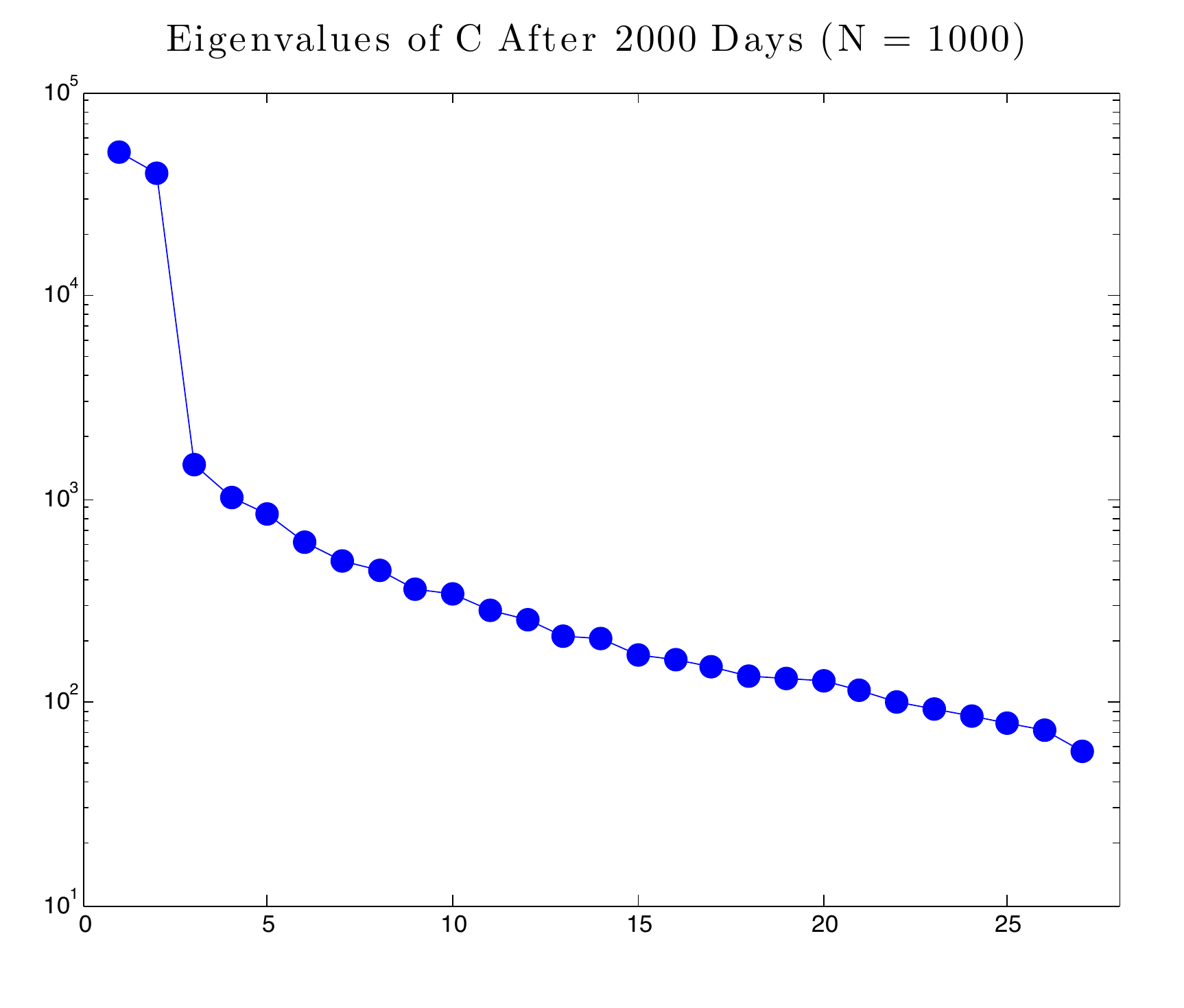}
\end{minipage}
\hspace{15mm}
\begin{minipage}{0.40\textwidth}
\centering
\hspace{-10mm}
\includegraphics[scale = 0.40]{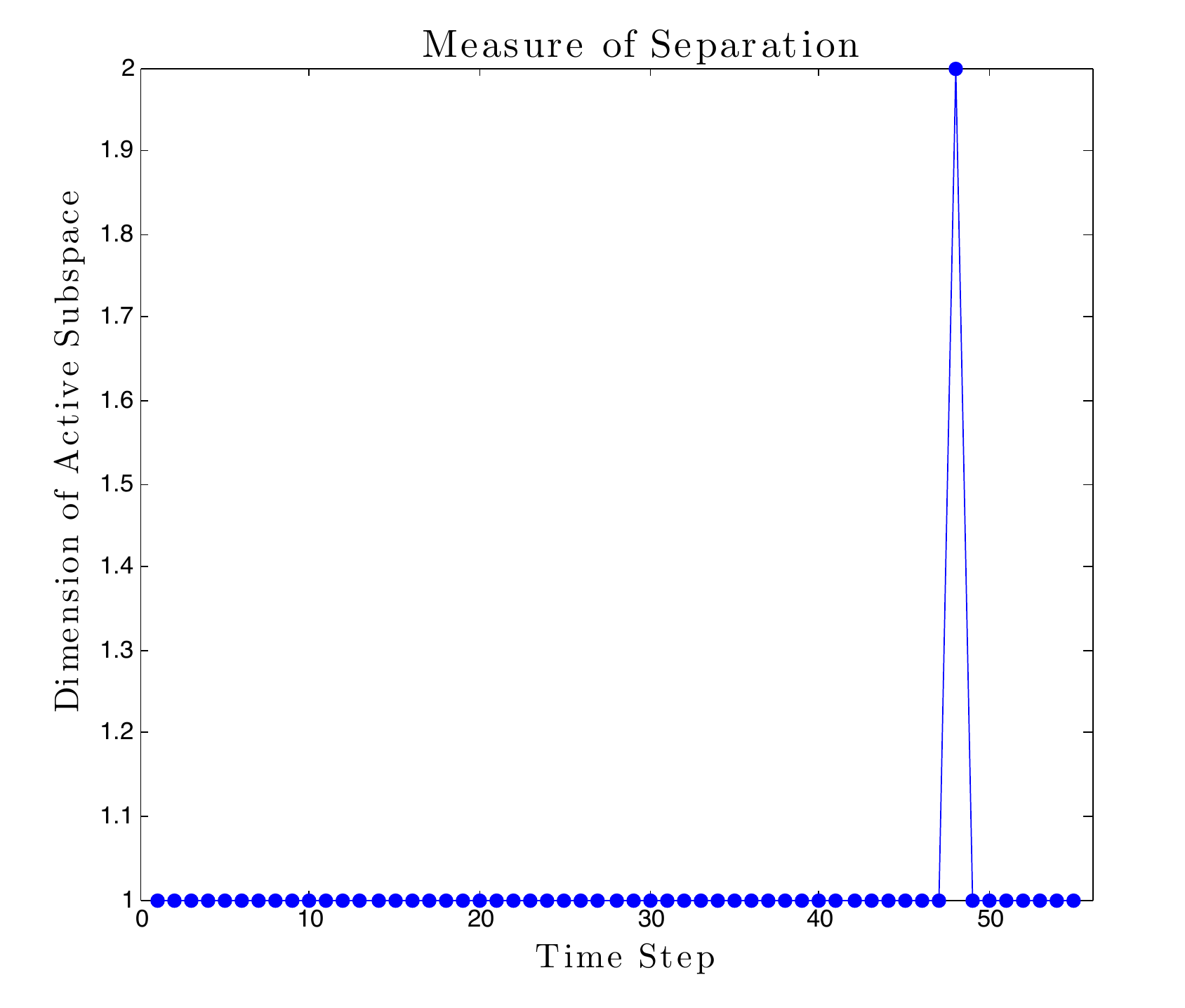}
\end{minipage}
\vspace{-0.2in}
\caption{Eigenvalues of the matrix $C$ after 2000 days (left). Dimension of the active subspace for each time (right).}
\label{fig:GlobalSep}
\vspace{-0.1in}
\end{figure}

The two weight vectors corresponding to these dominant eigenvalues are then computed and sufficient summary plots are obtained.
Figure~\ref{fig:SSP2D} shows the one and two dimensional sufficient summary plots, respectively, for the approximation of $T(2000)$. 
Considering the two dimensional sufficient summary plot in Figure~\ref{fig:SSP2D} (right), it can be seen that very little variation occurs within the T-cell count in the vertical direction, which represents the second active variable. All of the variation in the T-cell count appears to transpire mostly in the horizontal direction. Also, in Figure~\ref{fig:SSP2D} (left) we can see that the one dimensional sufficient summary plot clearly shows a distinct trend. For these reasons, the one-dimensional active subspace representation provides a sufficient description of the dynamics in the parameter space, while the second active variable doesn't appear to contribute as greatly.  Hence, even though the largest spectral gap appears between the second and third eigenvalues, a single linear combination of the parameters captures the overwhelming majority of the dynamics, and only the one-dimensional representation is utilized.  Figure~\ref{fig:GlobalSep} (right) displays the dimension of the active subspace at each time using the measure of separation and \eqref{dim}.

\begin{figure}[t]
\hspace{-24mm}
\centering
\begin{minipage}{0.4\textwidth}
\centering
\includegraphics[scale = 0.40]{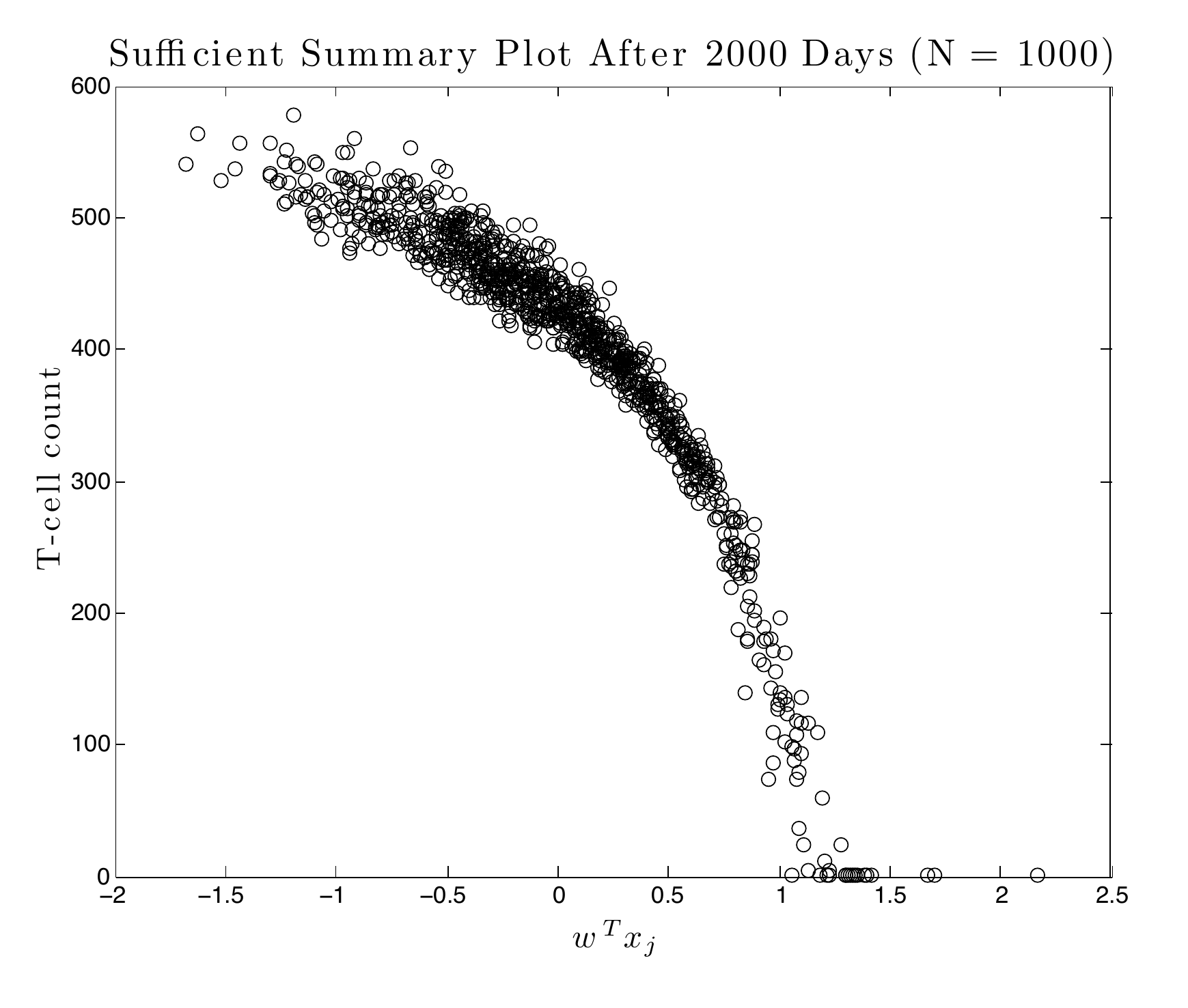}
\end{minipage}
\hspace{14mm}
\begin{minipage}{0.4\textwidth}
\centering
\hspace{-10mm}
\includegraphics[scale = .40]{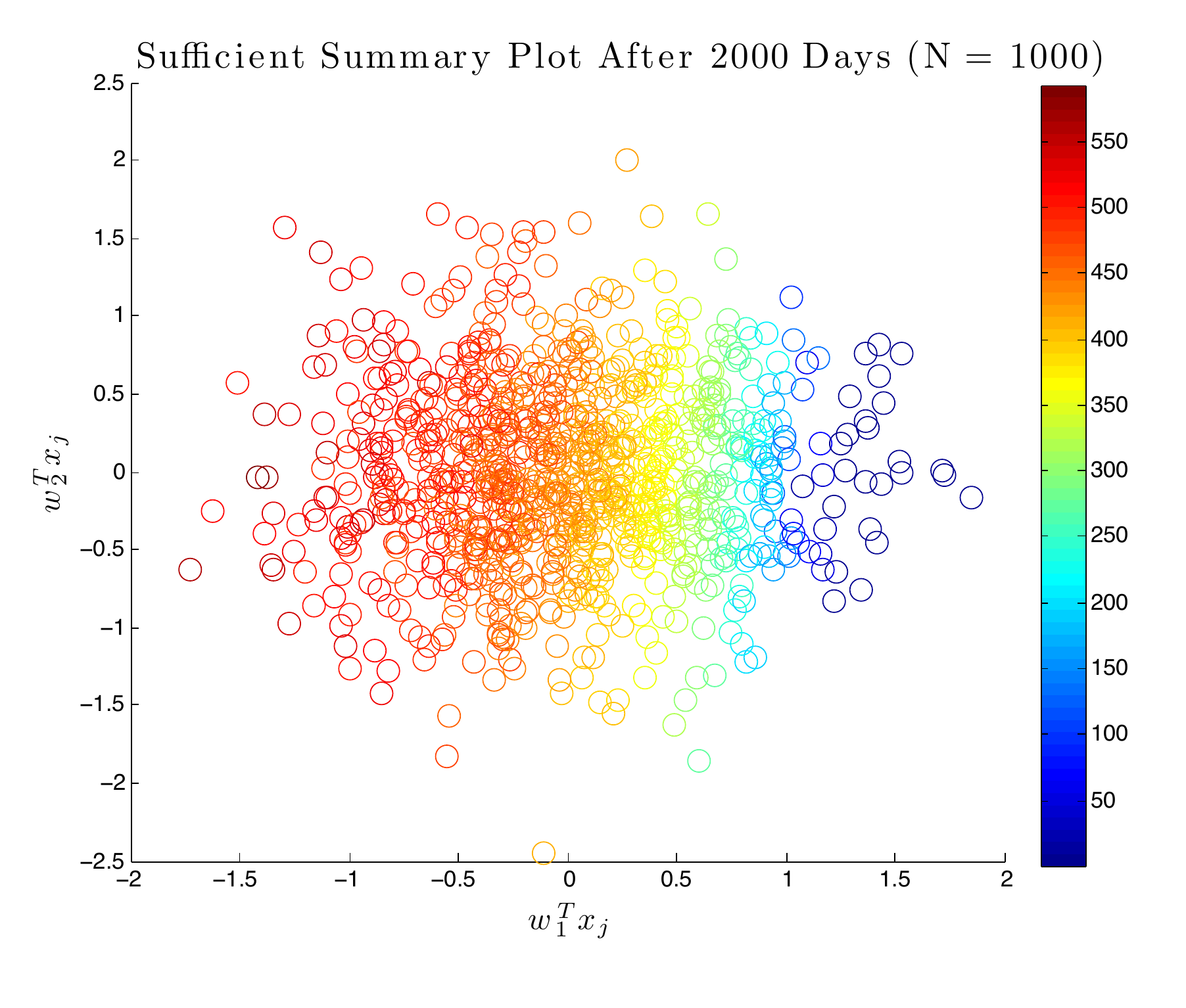}
\end{minipage}
\vspace{-0.2in}
\caption{Sufficient summary plots after $2000$ days, displaying the one-dimensional (left) and two-dimensional (right) active subspace representations}
\label{fig:SSP2D}
\vspace{-0.1in}
\end{figure}

Since this algorithm can be used to approximate the T-cell count at any fixed time, we can further compute active subspaces to reduce complexity in the parameter space and obtain a simple time course for $T(t)$.
We refer to such a model as a \emph{dynamic active subspace approximation}.  Within the next section we will demonstrate precisely how this reduced description can be constructed.

\subsection{Dynamic Active Subspaces}
Now that the approximation for $T(1700)$ has been computed, this method can be repeated for a predetermined discretization of the time domain in order to create a global-in-time approximation for the T-cell count, say $T(t; y)$. 
First, it is necessary to choose the mesh on which $T(t)$ will be approximated. In this case the time interval $[0, 3400]$ is divided into $86$ non-uniform subintervals and the active subspace decomposition, along with the associated eigenvalues and eigenvectors, is used to determine a functional approximation at each time step. 
Next, one must orient the eigenvectors to point in approximately the same direction so that they transition smoothly from one time step to the next. By this we mean that the magnitude of the components of the consecutive weight vectors differ only slightly, but because the normalized eigenvector decomposition is only unique up to a sign, weight vectors at different time steps must be oriented so that they do not point in opposing directions.
Upon correctly orienting the weight vectors, nonlinear curve-fitting is employed to compute the functions that best fit the associated sufficient summary plots. 
   
Once the approximations have been computed at each time step, the final task is to dynamically assemble them using linear basis functions, namely
\begin{equation}
\label{Tcell}
T(t; \mathbf{x}) \approx \sum_{i = 1}^{86} T_i \left (\mathbf{x} \cdot \mathbf{w}(t) \right )\phi_i(t)
\end{equation}
where $\mathbf{w}(t) = \sum \limits_{i=1}^{86} \mathbf{w}_i \phi_i(t)$ for all $t \in [0,3400]$ is a linear interpolation of the first active variable vector, $\mathbf{w}_i$, within the $i^{th}$ time interval, $y(t) = \mathbf{x} \cdot \mathbf{w}(t)$ is the corresponding active variable, $T_i(y)$ is the approximation to the T-cell count within the $i^{th}$ time interval, and $\phi_i(t)$ is the hat function on the interval $[t_{i-1},t_{i+1}]$ given by 
$$\phi_i(t) = \left \{
	\begin{aligned}
	\frac{t - t_{i-1}}{t_i - t_{i-1}}, & \quad t_{i-1} \le t \le t_{i} \\
	\frac{t_{i+1} - t}{t_{i+1} - t_i}, & \quad t_i \le t \le t_{i+1} \\
	0,\hspace{5mm} & \quad t \notin [t_{i-1},t_{i+1}].
	\end{aligned}
	\right.
$$
Recall that the original parameter values $\mathbf{p}$ and the normalized parameters $\mathbf{x}$ can be interchanged using \eqref{p}. Here we use the notation $T(t; \mathbf{x})$ instead of $T(t; y)$ because within the global-in-time approximation, the active variable vector, $\textbf{w}$ and hence the active variable $y$, changes within differing time intervals.  Of course, a higher-order approximation can be obtained by expanding $\phi_i(t)$ in a polynomial basis of greater degree.

\begin{figure}[t]
\vspace{-0.3in}
\hspace{-34mm}
\centering
\begin{minipage}{0.40\textwidth}
\centering
	\includegraphics[scale=.4]{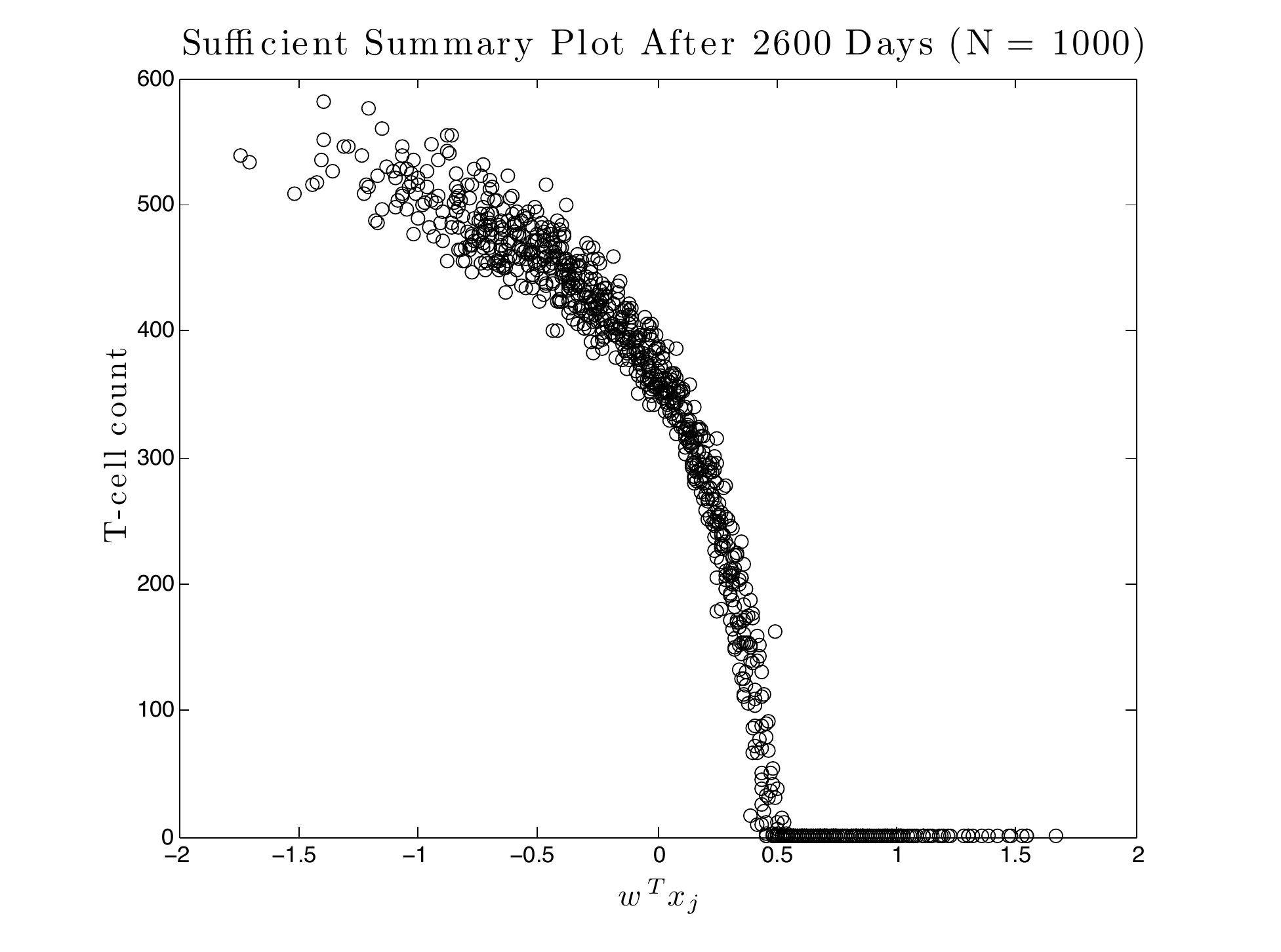}
\end{minipage}
\hspace{10mm}
\begin{minipage}{0.40\textwidth}
\centering
\hspace{-10mm}
	\includegraphics[scale=.4]{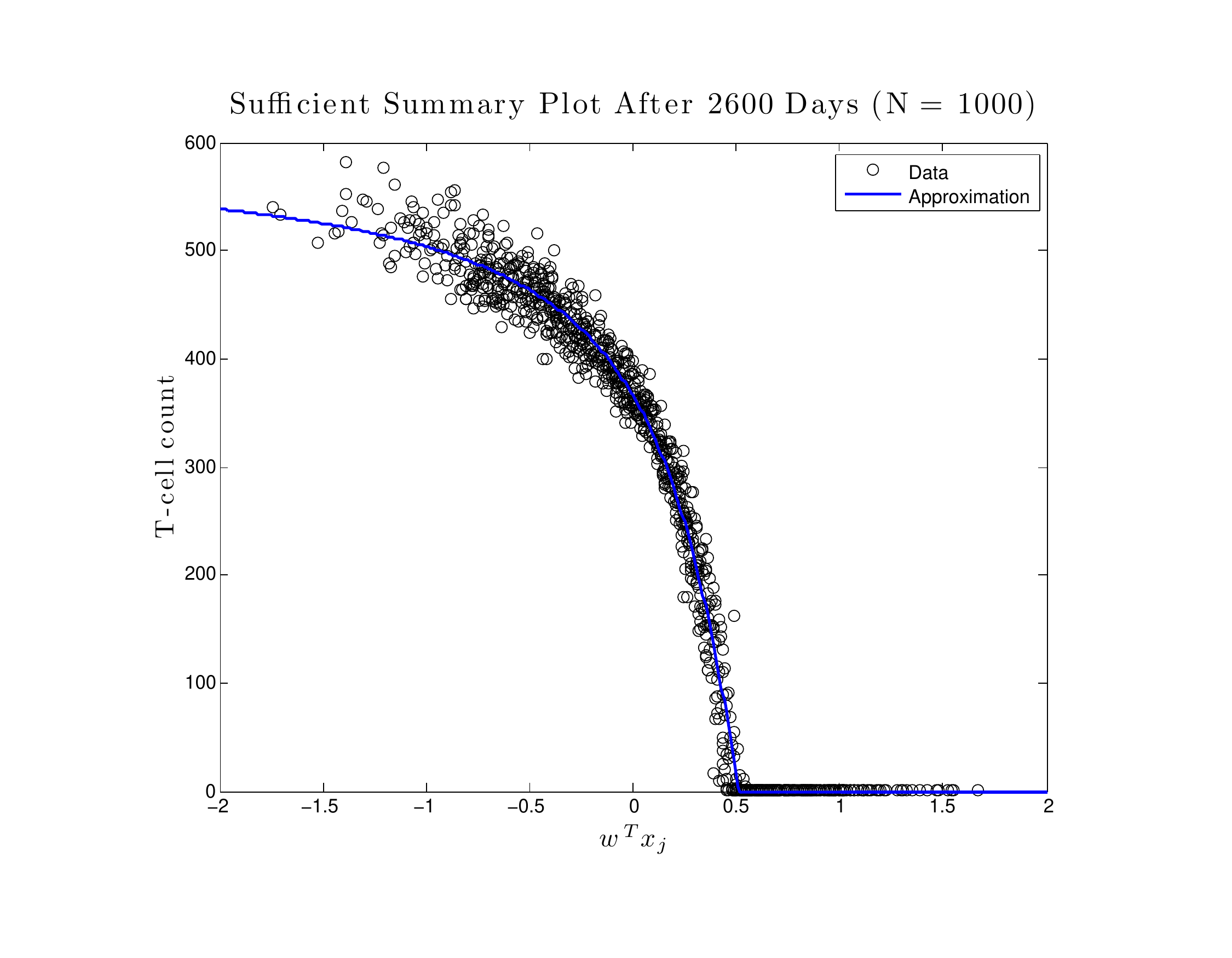}
\end{minipage}
\vspace{-0.35in}
\caption{Sufficient summary plot after 2600 days using 1000 trials (left). Same plot with function approximation (right).}
\label{fig:SSP2600}
\end{figure} 

The results of computing the active subspace and sufficient summary plot for $18$ of the $86$ time steps can be found in Appendix B.
With this we see that trends in the sufficient summary plot transition smoothly from one time step to the next and at each fixed time $t$ the T-cell count can be represented by one of three distinct functional forms. 
In particular, curves that are fit to the data generated within sufficient summary plots transition from (1) linear to (2) arctangent trends at or around $55-65$ days after initial transmission, while for times greater than or equal to $1800$ days, the trend resembles (3) an arctangent function multiplied by a heaviside step function.
The last of these aproximations occurs during later periods of the time course when the T-cell count may vanish for certain values of the active variable.  For example, Figure~\ref{fig:SSP2600} displays a sufficient summary plot of the T-cell count at $2600$ days after infection occurs.  For large values of the first active variable, say $y \geq 0.5$, we see that $T(2600; y) = 0$ since a proportion of sample runs result in a patient's T-cell count tending to zero prior to $2600$ days.  Contrastingly, for $y < 0.5$, the T-cell count trend resembles an arctangent function. 
These three specific functional approximations and their transitions comprehensibly display the biological aspects of all stages of infection dynamics.  Before discussing this further, a minor discussion of HIV disease pathogenesis is needed.

The time-course of HIV infection is characterized by three distinct stages: acute (or primary) infection, chronic infection (also referred to as the clinical latency stage), and the transition to AIDS. The first of these phases takes place within $10-20$ weeks of initial introduction of the virus within a host and is characterized by a rapid fluctuation in the T-cell and virion population. With respect to the T-cell population there is initially a rapid decrease from the introduction of the virus, and then a rapid rebound arising from the body's immune response. Symptoms during this phase of the infection include fever, swollen glands, fatigue, rash, and sore throat. The next stage, chronic infection, ranges from a number of years to over a decade without treatment. During the chronic period the T-cell and virion populations remain at relatively constant levels, with the T-cell population decreasing at a particularly slow rate and the number of virions increasing steadily. During the last stage of infection, the transition to AIDS occurs as the T-cell population reaches a density lower than 200 cells per mm$^3$.  The progression to AIDS is typically associated with a sharp decrease in the T-cell population within a year or so.

\begin{figure}[t]
\vspace{-0.1in}
\hspace{-15mm}
\begin{minipage}{0.3 \textwidth}
\centering
\includegraphics[scale = 0.33]{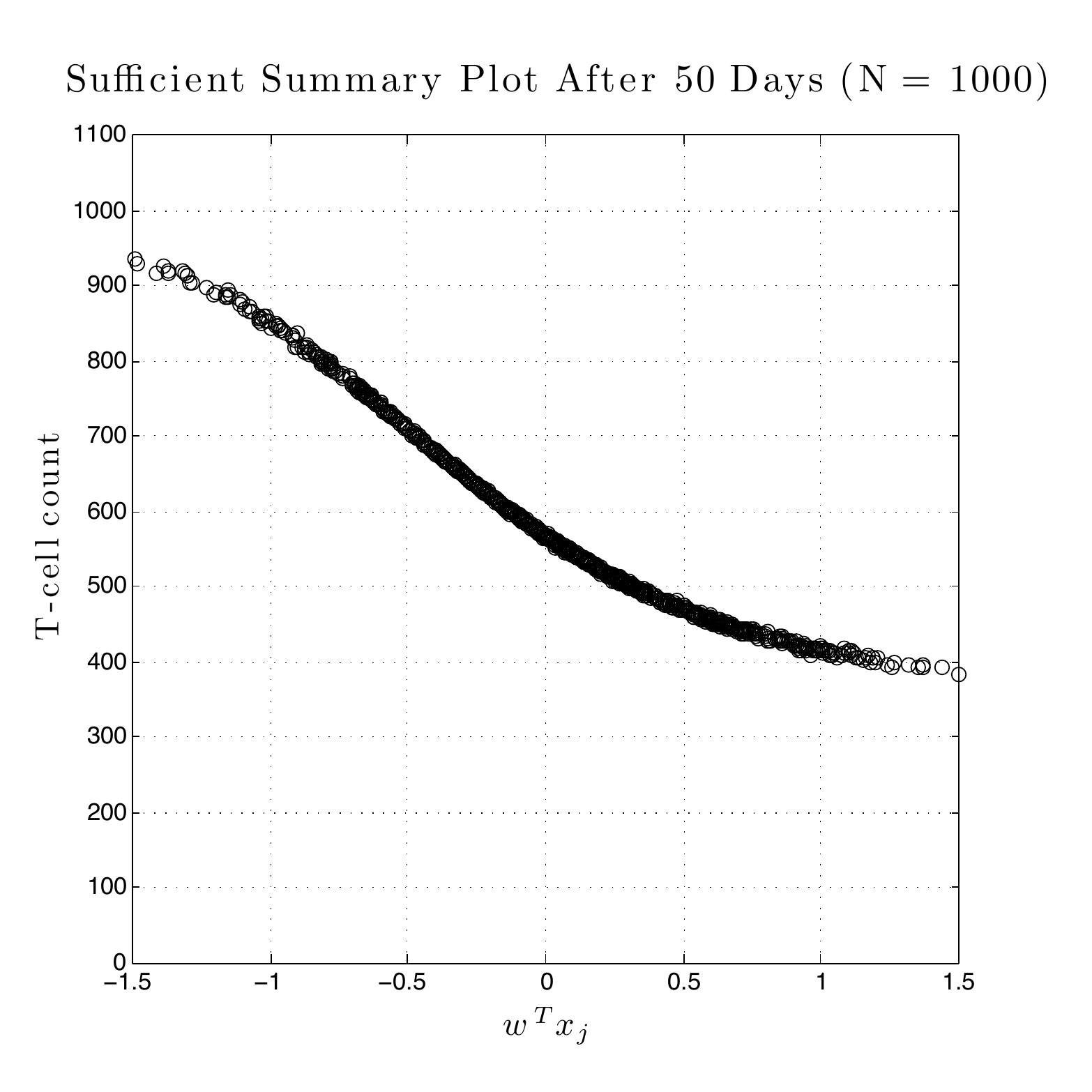}
\end{minipage}
\hspace{8mm}
\begin{minipage}{0.3 \textwidth}
\centering
\includegraphics[scale = 0.33]{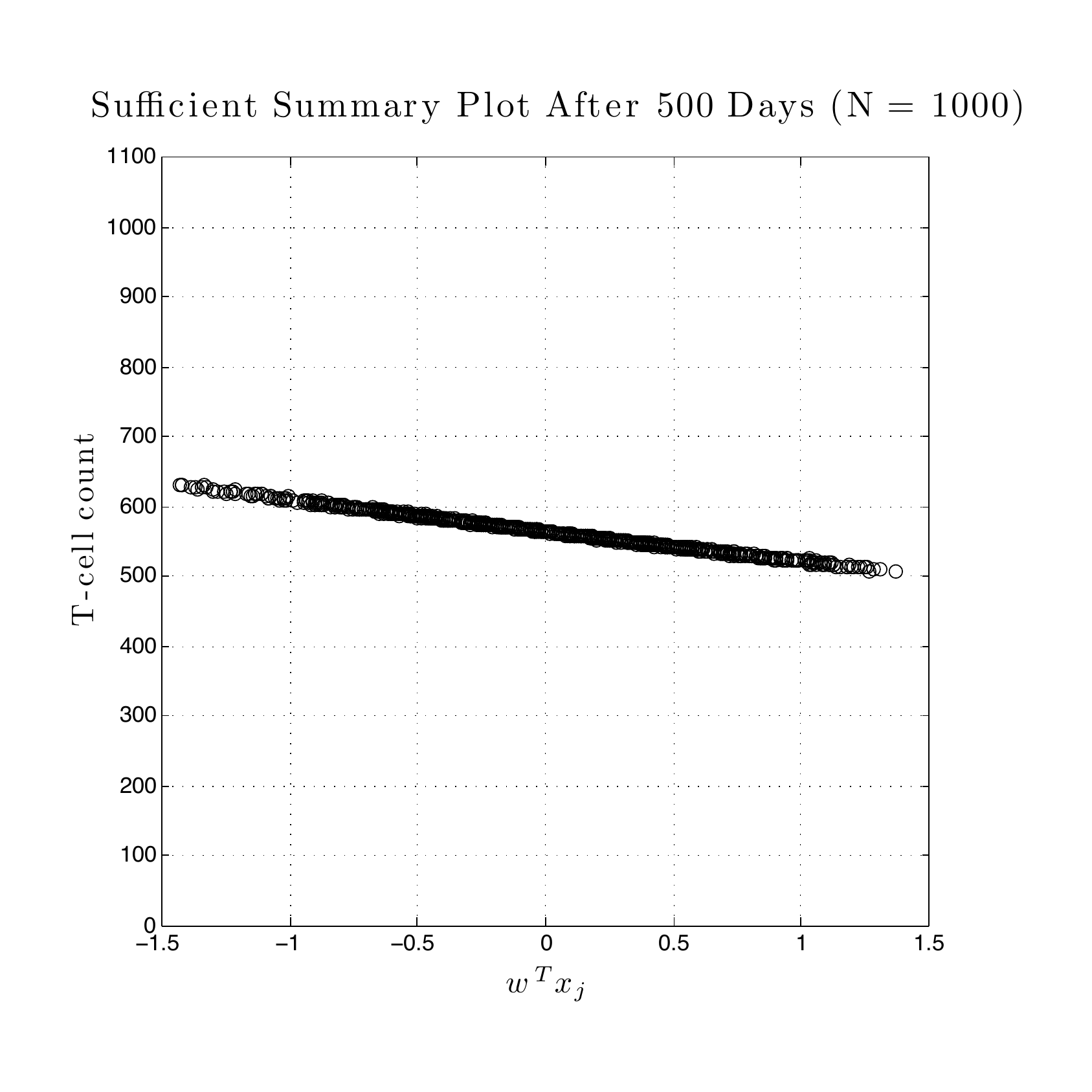}
\end{minipage}
\hspace{10mm}
\begin{minipage}{0.3 \textwidth}
\centering
\includegraphics[scale = 0.33]{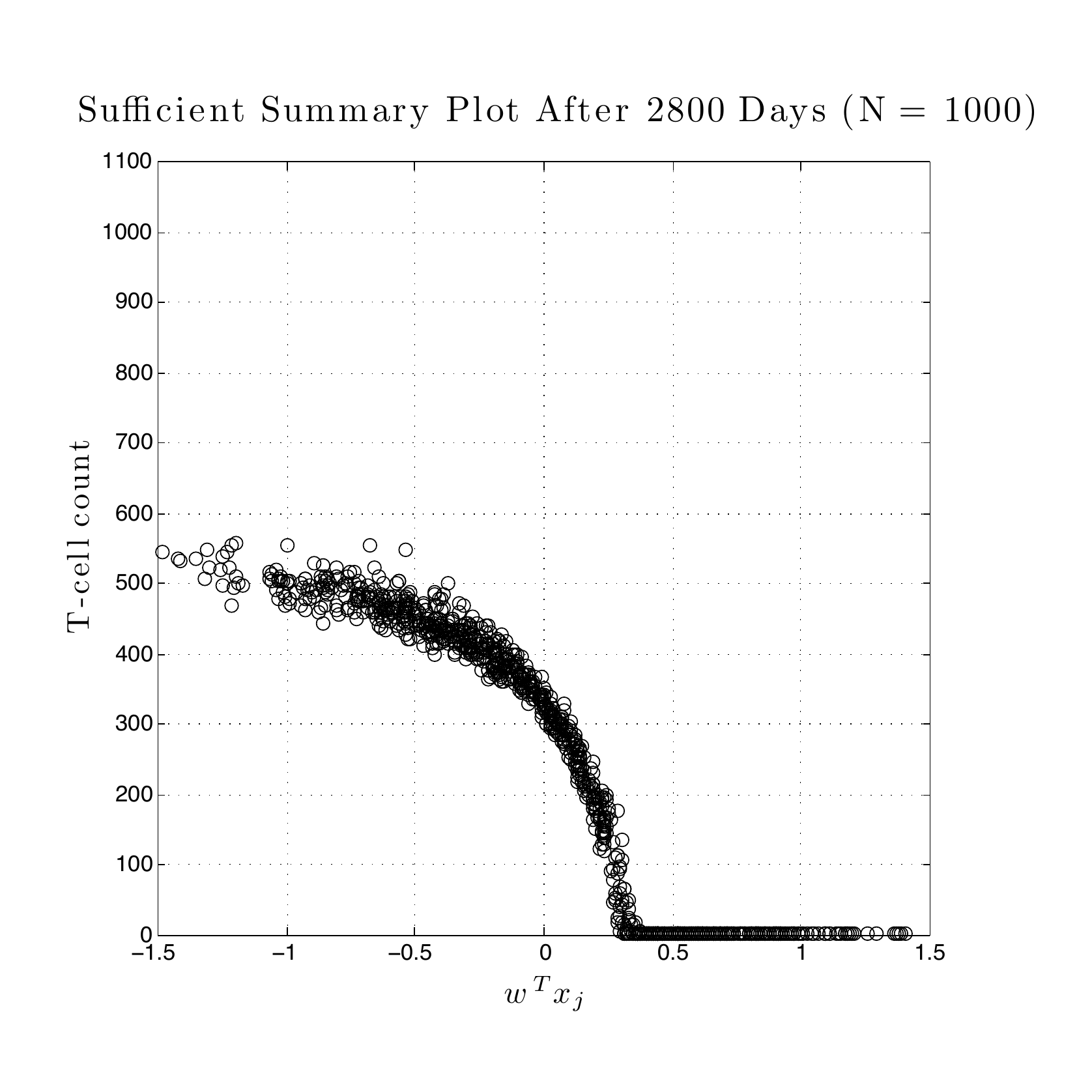}
\end{minipage}
\vspace{-0.3in}
\caption{Sufficient summary plots representing the three stages of infection - Acute (left), Chronic (center), AIDS (right)}
\label{fig:3stages}
\end{figure}

From the sufficient summary plots in Figure \ref{fig:3stages} and Appendix B, the three distinct stages of the disease become clear.
The three functional forms arising from sufficient summary plots precisely separate the three distinct stages of HIV disease progression within an infected individual - the initial arctangent function representing the acute phase, followed by a slow linear decline that denotes the asymptomatic or chronic phase, and finally the heaviside arctangent function detailing the decline of the T-cell count as a patient develops AIDS. 
This description of the T-cell population provides a detailed visual account of the three stages and the transitions between them.  
For clarity, a representative graph of each trend in the data is provided in Figure \ref{fig:3stages}.
Therefore, the abrupt changes in the T-cell count arising within the one-dimensional active subspace of parameters completely categorize the stages of disease pathogenesis.

In addition to separating these stages of the disease, the active subspace method allows for the creation of three different types of approximate models. First, a visual representation of the quantity of interest - in this case, the T-cell count - as a function of the most pertinent linear combination of parameters within the original model is provided.  Additionally, the method allows for the construction of an explicit, analytic model by combining nonlinear function approximations such as \eqref{analytic} over an interval of time.  In this direction a second approximation method is available instead of utilizing basis functions as in \eqref{Tcell}.  Namely, one may prescribe the functional form during a particular stage of the disease and fit time-dependent coefficients to the transitions within sufficient summary plots.  For instance, one may express the arctangent approximation over the timespan $t \in [0,55]$ by
$$T(t; y) = a(t) +b(t) \ {\rm tan}^{-1}(c(t)y  - d(t))$$
for well-behaved functions $a,b,c$, and $d$ that are fit to the changes in the data.  This type of explicit approximation would be easiest to implement during the asymptomatic phase for which 
$$T(t; y) = m(t)y + b(t)$$
and the slope and $T$-intercept functions, $m(t)$ and $b(t)$ respectively, are given by the piecewise-defined or smoothed functions approximated in Figure \ref{linfits}.

\begin{figure}[t]
\vspace{-0.1in}
\hspace{-20mm}
\begin{minipage}{0.4\textwidth}
\centering
	\includegraphics[scale=.4]{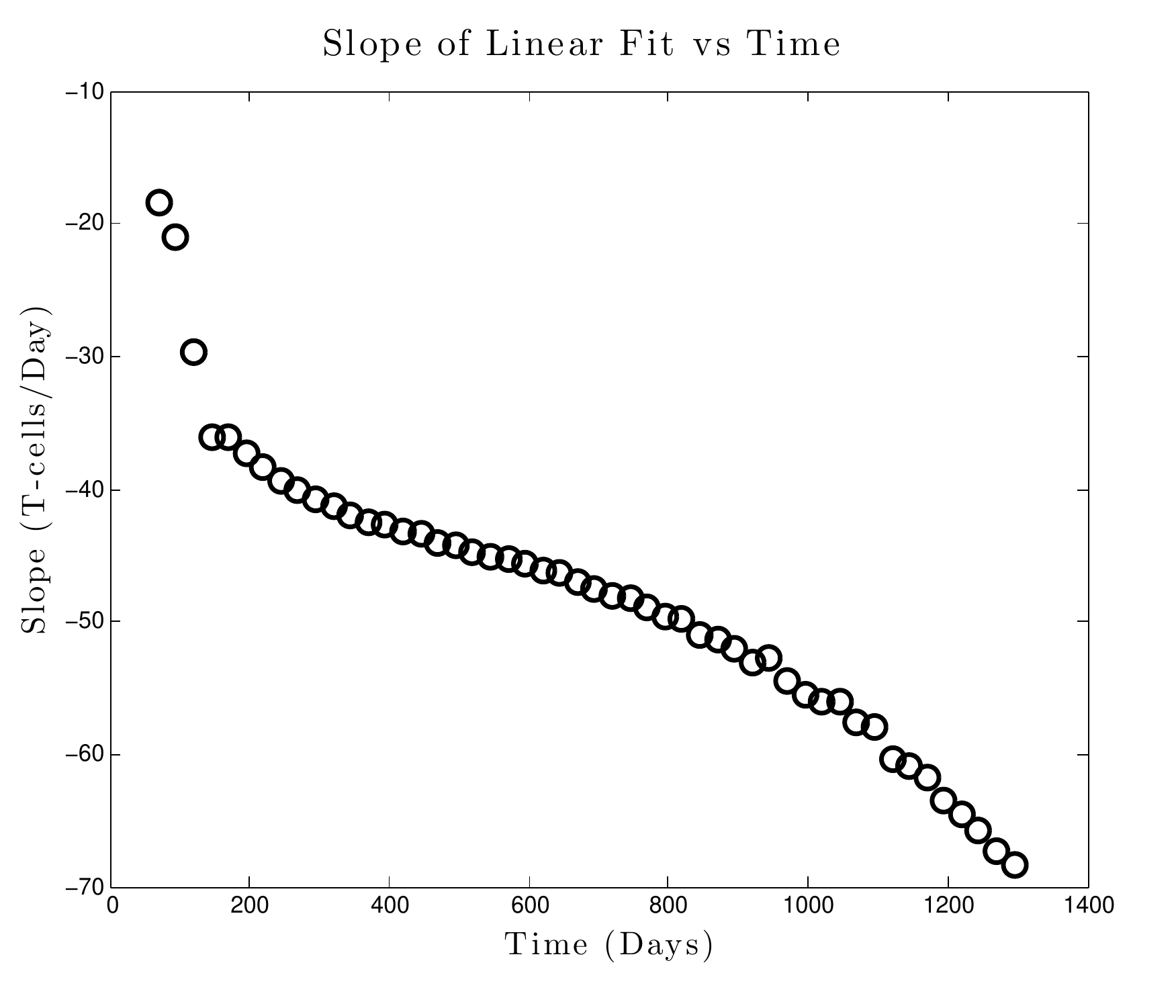}
\end{minipage}
\hspace{1.3cm}
\begin{minipage}{0.4\textwidth}
\centering
\hspace{-10mm}
	\includegraphics[scale=.4]{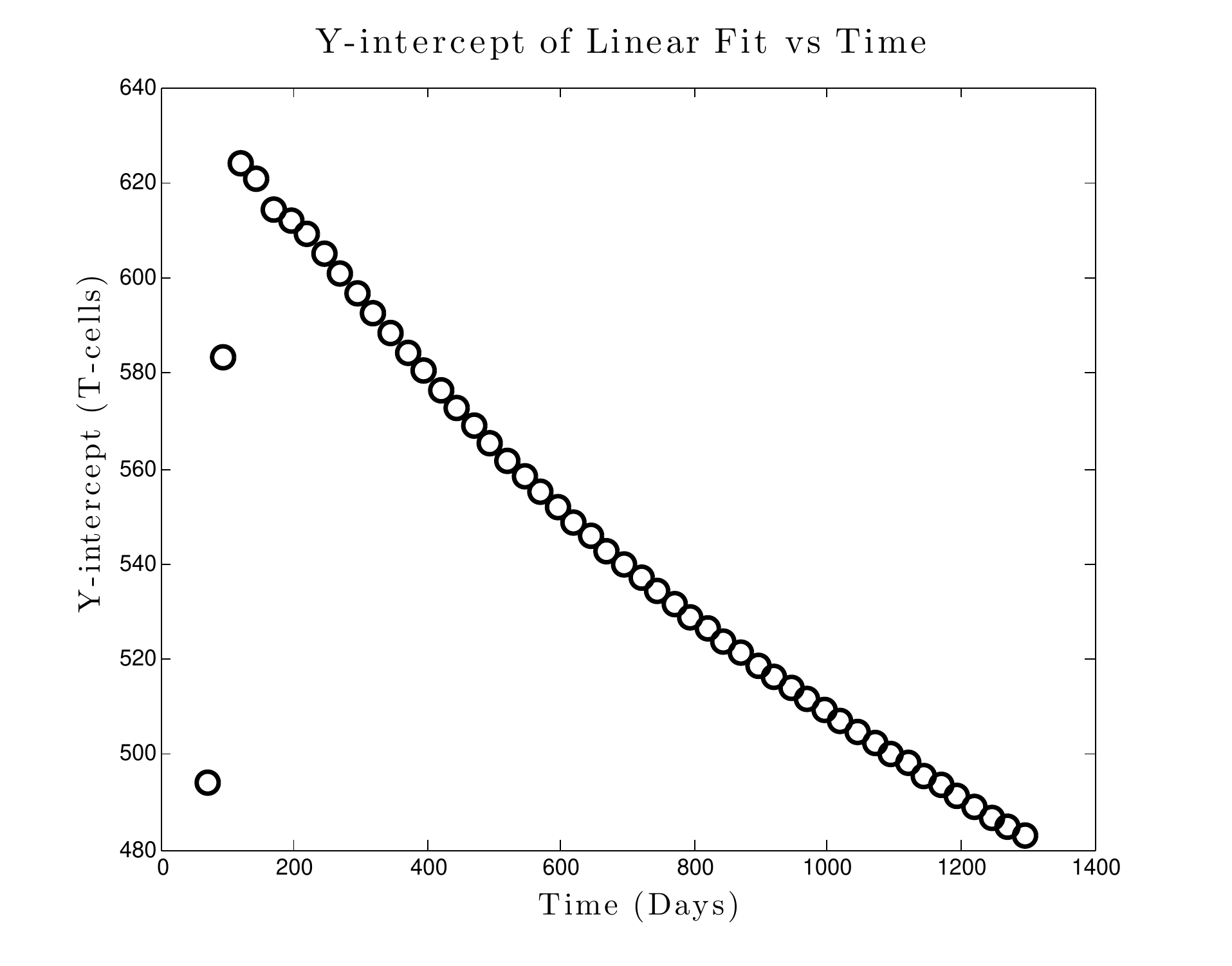}
\end{minipage}
\vspace{-0.2in}
\caption{Slope (left) and $T$-intercept (right) functions, $m(t)$ and $b(t)$, respectively for $t \in [55, 1300]$.}
\label{linfits}
\end{figure}

Lastly, the nonlinear fits arising from the sufficient summary data can be easily stored and supply the basis for a low-cost computational approximation without the need to simulate the original model over long time periods.  In contrast to simulating the full system of ODEs for each new set of parameter values, this computational model need only be precomputed once and can easily describe which of the parameters are most important and during which stages they significantly alter the biological quantities of interest.
As parameters within the original HIV model \eqref{Hadji} vary amongst patients, they must be recomputed for each individual, and the computational savings provided by the reduced model is vital.

Utilizing the dynamical algorithm represented by \eqref{Tcell}, this computational model can be constructed and compared with a representative simulation of the full dynamical system. Figure~\ref{fig:Globalapprox} displays both the simulated T-cell count and its global-in-time approximation using dynamic active subspaces.  We note that because the active subspace method is global with respect to the parameter space, the precise values of parameters in Table \ref{tab:paramvalues} do not necessarily influence the structure of the active subspace model, as long as a feasible range of parameter values is available.
%



In order to test the accuracy of the active subspace approximation to solutions of \eqref{Hadji}, $100$ independent simulations were performed and the relative error was computed. Within these error calculations a uniformly-distributed random time 
was selected along with a uniformly-drawn selection of the parameter values within their respective ranges. The result is shown in Figure~\ref{fig:GlobalError}, and displays that for $96\%$ of simulations the analytic approximation varied less than $5\%$ from the value given by the stiff differential equation solver.  Hence, one may conclude that the global-in-time active subspace model well-approximates solutions to the original system of ODEs given in \eqref{Hadji}.

\begin{figure}[t]
	\vspace{-0.1in}
	\hspace{-10mm}
	\includegraphics[scale=.5]{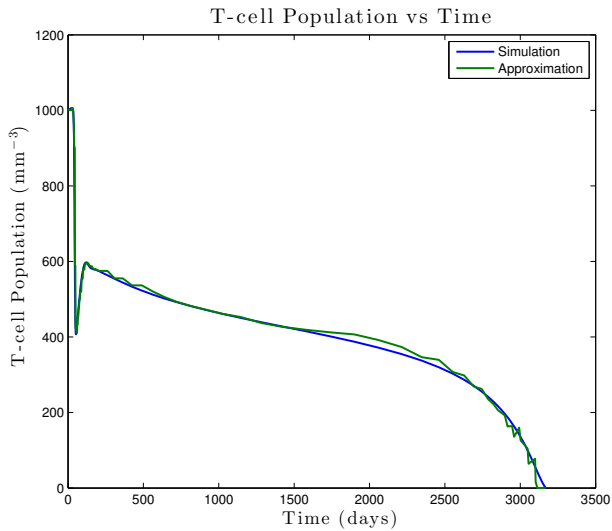}
	\vspace{-0.1in}
	\caption{Global-in-time approximation of the T-cell count.}
	\vspace{-0.1in}
	\label{fig:Globalapprox}
\end{figure}

\subsection{Dimension Reduction in the Parameter Space}
As determined in the previous section, not all parameters are required to describe the behavior of solutions within the model and those that are needed may not be important during each stage.
Therefore, it makes sense to investigate the construction of reduced models, namely those which eliminate the contributions of certain parameters.
Though the active subspace method has reduced the dimension of the parameter space upon which the T-cell count depends, all of the parameter values are still needed in order to compute the approximate solution, as the reduced parameter space has been expressed merely as a linear combination of the original parameter values, i.e. $y = \mathbf{x} \cdot \mathbf{w}$.
However, the weights (or coefficients) within this linear combination, given by $\mathbf{w}$, should provide a clear method to reduce the original parameter space as well.  For instance, assume that a parameter, say $x_3$, possesses a corresponding weight entry $w_3$ that is very small in comparison to the other weights.  Then, variations in $x_3$ will have little to no effect on the reduced parameter space $y$, and hence will not appreciably influence the output variable $T(t; y)$. Thus, one can merely eliminate $x_3$ within the model or set $x_3 \equiv 0$ rather than considering it as a parameter whose value could possibly vary depending upon the patient.

\begin{figure}[t]
\centering
	\includegraphics[width = 0.6\textwidth]{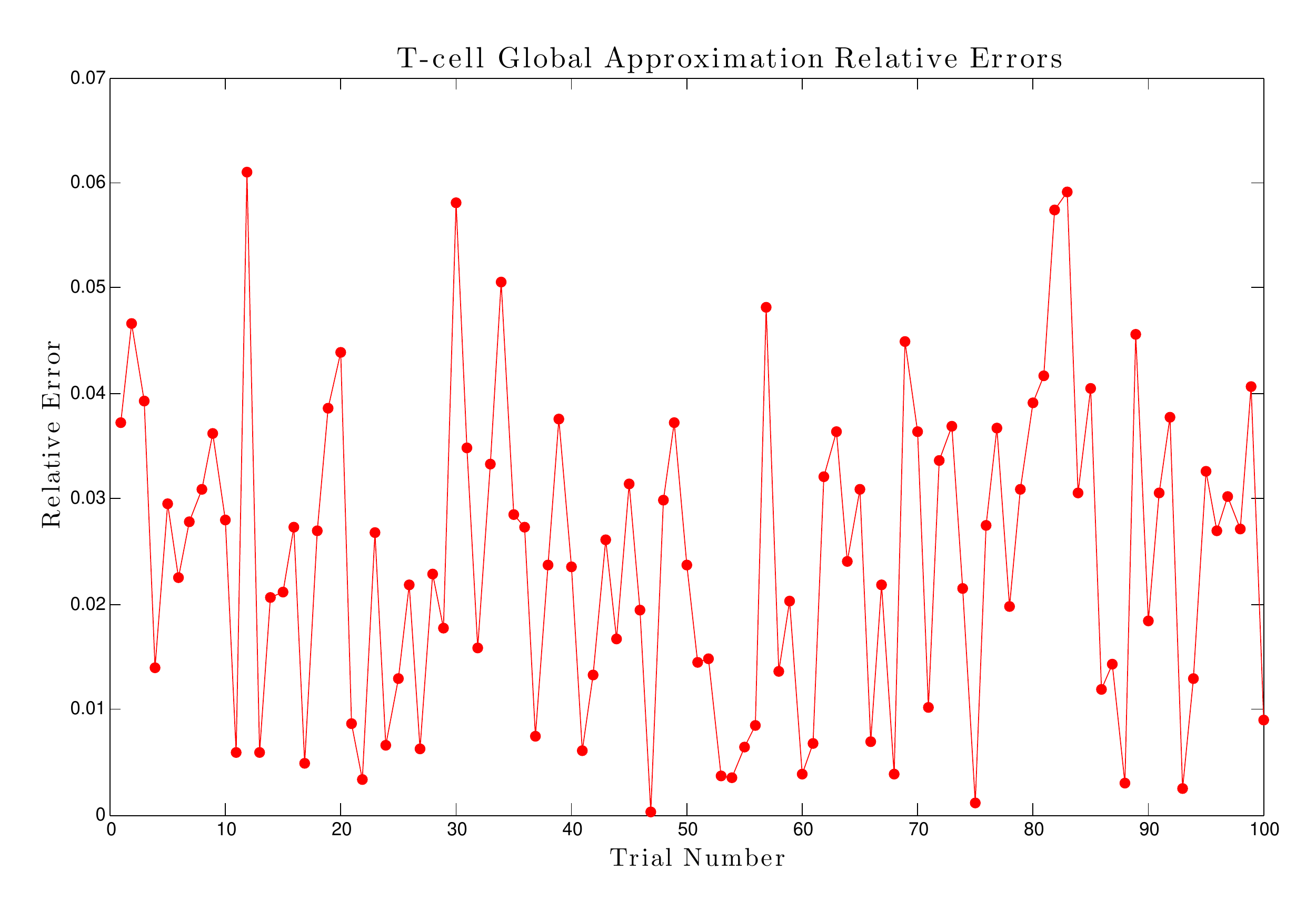}
	\vspace{-0.2in}
	\caption{Relative error in the global approximation of the T-cell count.} 
	\label{fig:GlobalError}
\end{figure}  

While it would be most beneficial to determine those weight vectors that remain small throughout the entire course of infection and remove the corresponding parameters, this is an unlikely scenario as different parameters will typically influence different stages in some substantial manner. For this reason, a true global-in-time parameter reduction is highly unlikely.  Indeed, the computed weight vectors for this study show that at most $3$ of the $27$ parameters - namely, $K_{11}$, $K_{12}$, and $K_{13}$ - could be eliminated without introducing enormous variations in the behavior of the model during some stage of infection.  Such a limited reduction in complexity would be marginally beneficial to the expense of computations.
Instead, we may separate the dynamics into the three distinct stages of the disease and use parameter reduction to create new models for each stage separately.  By eliminating the associated interactions from \eqref{Hadji} we may derive a simpler system that still accurately predicts each of the three phases of infection.

As an illustrative example, we consider the weight vectors arising within the first $40$ days of initial infection and remove parameters whose weights remain below a fixed threshold throughout this interval of time.  Using a relatively strict threshold, in this case $0.032$, a number of parameters are removed from \eqref{Hadji}, namely $s_2$, $s_3$, $K_2$, $K_4$, $K_5$, $K_6$, $K_7$, $K_8$, $K_{10}$, $K_{11}$, $K_{12}$, $K_{13}$, $\delta_3$, $\delta_4$, $\delta_5$, $\delta_6$, and $\alpha_1$. 
With these parameters eliminated, the $T_L$, $M$, $M_I$, and $CTL$ populations decouple from the remaining equations.
This implies that latently-infected T-cells, macrophages (both healthy and infected), and cytotoxic lymphocytes do not play an important role in the early behavior of the disease, i.e. within the first five to six weeks of introduction of the virus within the body.   Because these populations decouple, it's necessary to remove $K_3$ (whose maximum weight during the first $40$ days is $0.1256$) and set $\psi$ = 1, thereby yielding the reduced system for the acute stage, namely
\begin{equation}
\label{ShuttPankavich}
\left.
\begin{aligned}
\frac{dT}{dt} &= s_1 + \frac{p_1}{C_1+V}TV - \delta_1 T - K_1 TV \\
\frac{dT_I}{dt} &= K_1 TV - \delta_2 T_I \\
\frac{dV}{dt} &= K_9 T_I - \delta_7 V. 
\end{aligned}
\right \}
\end{equation}
This model is an augmented form of the well-known three-component model (see \cite{Roemer}) with an additional Michaelis-Menten term within the T-cell population, which accounts for the homeostatic proliferation of such cells upon depletion of this compartment due to interactions with virions. This model was recently analyzed in detail in \cite{PS2} and found to possess many desirable properties, including bistable equilibria which explain the dependence of infection dynamics on the initial T-cell count and viral load, as well as the existence of a Hopf bifurcation which describe oscillations within the system.
Using the fitted parameter values given in \cite{PS2} for the model \eqref{ShuttPankavich} and plotting against the full model \eqref{Hadji} with the standard parameter values given in Table \ref{tab:paramvalues} results in Figure \ref{fig:ReducedModel}.

Hence, for the acute stage we have reduced \eqref{Hadji} with $27$ independent parameters, to \eqref{ShuttPankavich}, which features only $8$ parameters. Following the same procedure for the other two stages would further reduce the dimension of the parameter space in the model \eqref{Hadji} at later stages of the infection.
These reduced models could then be used to analyze the disease separately within each distinct stage of infection and with much less computational cost.  
If a global model is still preferred in comparison to separating models by stage, one can perform this parameter reduction for all three stages subject to the constraint that values of $T(t)$ within time-adjacent models be equal and a new global-in-time dynamical model arises by construction.  Such a reduced, long-term model would be comparable to that of Figure \ref{fig:Globalapprox}. 

\begin{figure}[t]
\centering
	\includegraphics[width = 0.66\textwidth]{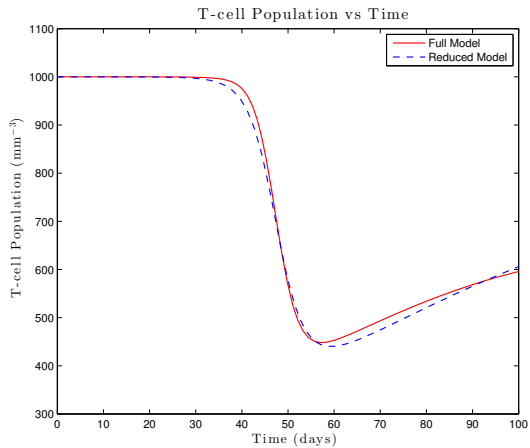}
	\vspace{-0.2in}
	\caption{Full HIV model versus reduced HIV model for the first 100 days.  Parameter values within the reduced model are $s_1 = 10$, $p_1 = 0.2$, $C_1 = 55.6$, $\delta_1 = 0.01$, $K_1 = 4.72 \times 10^{-3}$, $\delta_2 = 0.69$, $K_9 = 5.37\times 10^{-1}$, and $\delta_7 = 2.39$ } 
	\label{fig:ReducedModel}
\end{figure}  

\section{Conclusion}

The current study concerns an analysis of the system \eqref{Hadji}, which is one of the only mathematical models to accurately represent all three stages of HIV infection within a host. A unique, biologically-relevant virus-free equilibrium was shown to exist, and conditions were determined that guarantee the local asymptotic stability of this state. Then, using dynamic active subspaces, the system of seven ODEs with $27$ parameters was approximated by algebraic and computational models while retaining the majority of pertinent information and system behavior. This method enabled the discovery of a dominant subspace of parameters, within which the T-cell count is most sensitive to perturbations, and allowed us to perform thorough parameter studies in this direction, rather than the complete $27$-dimensional parameter space.
In general, active subspace methods also provide for a deeper visualization of the dependence of solutions on the parameter space, as well as, an analytic model \eqref{Tcell} to describe the quantity of interest and a computational model for which solutions are less expensive to construct. The efficacy of the model \eqref{Tcell} was investigated by calculating the relative error compared to solutions of \eqref{Hadji} solved with a stiff ODE solver.

As with any study, a number of future questions arise from our investigation.
First, as many parameter values within \eqref{Hadji} are not specifically known and vary greatly in the literature, one would like to quantify the uncertainty inherent in choosing a particular value for each.  While this can be investigated using the approximations established herein, it was not the focus of the study and more complex mathematical tools are needed to do so. 
Another direction to consider is the explicit quantification of error within the lower-dimensional approximations provided by active subspaces.  Some preliminary results appear in \cite{Aspaces}, but exact bounds and convergence theorems for dynamic, rather than time-independent, active subspaces are currently unavailable within the literature.
Finally, we note that these methods can be used for other in-host models of HIV (or other physical and biological models \cite{CZC}) that possess high-dimensional parameter spaces.  For instance, a recent refinement of \eqref{Hadji} was proposed in \cite{HVM}, and appears to display less sensitivity to variations in parameter values.  Hence, a similar technique would likely be useful to conduct a parameter study or construct a reduced-order model for this system.






\section*{Acknowledgments} The authors would like to thank Prof. Paul Constantine for helpful discussions and advice.


\appendix

\section{Proofs of Theorems}
\noindent In the first appendix, we outline the proofs of the theorems stated in Section~\ref{ch:model}.
\begin{proof}[Proof of Theorem \ref{T1}]
\label{proof1}
Beginning with (\ref{Hadji}), we search for steady states by assuming that all time derivatives are zero within the equations, and attempt to solve for the constant states $(T, T_I, T_L, M, M_I, CTL, V)$.  Assuming $V=0$ within the system of ODEs provides a significant reduction in the complexity
of the system.  The $M$ equation implies $M = \frac{s_2}{\delta_4}$.  Using this within the equation for $M_I$ implies that either $M_I = 0$ or $CTL = -\frac{\delta_5}{K_6}$.

Consider the latter case first. Multiplying the $T_I$ equation by $(1 - \psi)$ and the $T_L$ equation by $\psi$ and adding gives
\begin{equation}
\label{blah}
 0 = (\alpha_1 + \psi \delta_3)T_L + \Bigg( \frac{(1 - \psi)K_3 \delta_5}{K_6} - (1 - \psi) \delta_2 \Bigg) T_I.
\end{equation}

Creating a linear system with \eqref{blah}, the $CTL$ equation, and the $V$ equation then solving for $T_I$, $T_L$, and $M_I$ gives
$$T_I = K_{10} \omega, \qquad
T_L = \frac{K_{10} \xi \omega}{K_6 (\alpha_1 + \delta_3 \psi)}, \qquad
M_I  = -K_9 \omega$$
where $$ \omega = \frac{s_3 K_6 + \delta_5 \delta_6}{\delta_5 (K_7 K_{10} - K_8 K_9)} \quad \rm{and} \quad \xi = (1 - \psi)(\delta_2 K_6 - \delta_5 K_3).$$ 
Lastly, inserting the value of $M_I$ into the $T$ equation and solving for $T$ yields
\begin{equation*}
T = \frac{s_1}{\delta_1 - \omega K_2 K_9}
\end{equation*}

Finally, consider the former case.  Then, it follows from the equation for $V$ that $T_I = 0$ as well. Collecting these terms in the $CTL$ differential equation implies that $CTL = \frac{s_3}{\delta_6}$. The equations for $T_I$ and $T_L$ together imply $T_L = 0$, and finally, with the remaining populations determined, the first equation implies $T = \frac{s_1}{\delta_1}$.  Hence, we find the steady state
$$ E_{NI} := \left (\frac{s_1}{\delta_1}, 0, 0, \frac{s_2}{\delta_4}, 0, \frac{s_3}{\delta_6}, 0 \right ).$$

We note that assuming all parameter values are strictly positive implies that the only non-infective steady state of biological significance is $E_{NI}$.
\end{proof}

Finally, we sketch the proof of the asymptotic stability result, which utilizes standard methods from the theory of dynamical systems (i.e. the Hartman-Grobman and Routh-Hurtwitz theorems) to determine the qualitative behavior of the $E_{NI}$ steady state from (\ref{Hadji}).

\begin{proof}[Proof of Theorem \ref{T3}]
As the model is not positivity-preserving a simple technique like the next-generation method does not apply.  Hence, we will utilize the Hartman-Grobman Theorem to arrive at the stated result and omit some of the more technical details.
We begin by computing the Jacobian of (\ref{Hadji}) evaluated at the steady states $E_{NI}$

\begin{frame}
\footnotesize
\arraycolsep=2pt 
\medmuskip = 1mu 
$$ \footnotesize
J(E_{NI}) = \left(
\begin{array}{ccccccc}
 -\delta_1 & 0 & 0 & 0 & -\frac{K_2 s_1}{\delta_1} & 0 & \frac{\left(p_1-c_1 K_1\right) s_1}{c_1 \delta_1} \\
 0 & -\frac{\delta_2 \delta_6+K_3 s_3}{\delta_6} & a_1 & 0 & \frac{p K_2 s_1}{\delta_1} & 0 & \frac{p K_1 s_1}{\delta_1} \\
 0 & 0 & -a_1-\delta_3 & 0 & -\frac{(p-1) K_2 s_1}{\delta_1} & 0 & -\frac{(p-1) K_1 s_1}{\delta_1} \\
 0 & 0 & 0 & -\delta_4 & 0 & 0 & \frac{\left(K_4-K_5\right) s_2}{\delta_4} \\
 0 & 0 & 0 & 0 & -\frac{\delta_5 \delta_6+K_6 s_3}{\delta_6} & 0 & \frac{K_5 s_2}{\delta_4} \\
 0 & \frac{K_7 s_3}{\delta_6} & 0 & 0 & \frac{K_8 s_3}{\delta_6} & -\delta_6 & 0 \\
 0 & K_9 & 0 & 0 & K_{10} & 0 & -\delta_7-\frac{K_{11} s_1}{\delta_1}-\frac{\left(K_{12}+K_{13}\right) s_2}{\delta_4} \\
\end{array}
\right).
$$
\end{frame}

From this, we can see that three eigenvalues are certainly real and negative
$$\lambda_1 = - \delta_1, \quad \lambda_2 = - \delta_4, \quad \lambda_3 = - \delta_6.$$ 
The remaining four eigenvalues are more difficult to identify as they are determined by the quartic equation
$$ a_4 \lambda^4 + a_3 \lambda^3+ a_2 \lambda^2+ a_1 \lambda + a_0 = 0$$
where 
$$a _ 4 = \delta_1 \delta_4 \delta_6^2 > 0$$
\begin{eqnarray*}
a _ 3 & = &  \alpha_1 \delta_1 \delta_4 \delta_6^2+\delta_4 \delta_6^2 K_{11} s_1+\delta_1 \delta_6^2 K_{12} s_2+\delta_1 \delta_6^2 K_{13} s_2+\delta_1 \delta_4 \delta_6 K_3 s_3\\
& \ & +\delta_1 \delta_4 \delta_6 K_6 s_3+\delta_1 \delta_2 \delta_4 \delta_6^2+\delta_1 \delta_3 \delta_4 \delta_6^2+\delta_1 \delta_4 \delta_5 \delta_6^2+\delta_1 \delta_4 \delta_7 \delta_6^2
> 0
\end{eqnarray*}
and $a_0$, $a_1$, and $a_2$ are given by much longer expressions and not necessarily positive.

Instead, we must impose conditions on each term to guarantee positivity, which is needed for the roots of the quartic to possess negative real part by the Routh-Hurwitz criteria. 
In particular, the two negative terms in $a_2$ are dominated by the remaining positive terms if and only if
$$K_1K_9 \leq \delta_2K_{11}$$ and
$$K_5K_{10} \leq (K_{12} + K_{13} ) \delta_5.$$
The same conditions imply the positivity of $a_1$.  For $a_0$, the negative terms are dominated by positive terms if and only if the two previous conditions hold and 
$$ K_2 K_5 K_9 s_1 s_2 \leq \delta_1 \delta_2 \delta_4 \delta_5 \delta_7 + \delta_4 \delta_5 K_1 K_4 s_1 + \delta_1 \delta_2 K_5 K_{10} s_2.$$
The final inequalities of the Routh-Hurwitz criteria are also implied by these conditions.
Hence,  defining
$$R_1 = \frac{K_1K_9}{\delta_2K_{11}},$$
$$R_2 = \frac{K_5K_{10}}{(K_{12} + K_{13} ) \delta_5},$$
$$R_3 = \frac{K_2 K_5 K_9 s_1 s_2}{\delta_1 \delta_2 \delta_4 \delta_5 \delta_7 + \delta_4 \delta_5 K_1 K_9 s_1 + \delta_1 \delta_2 K_5 K_{10} s_2},$$
and $$R_0 = \max \{ R_1, R_2, R_3 \}$$
we see that the equilibrium is locally asymptotically stable if and only if all three conditions are satisfied, and thus $R_0 \leq 1$.
\end{proof}

\section{Visual representation of three stages of infection}
In Figures \ref{fig:SSPs1}, \ref{fig:SSPs2}, and \ref{fig:SSPs3}, we display a more temporally refined representation of the three stages of infection displayed by the model.

\noindent\begin{minipage}{\textwidth}
\vspace{0.2in}
\hspace{-13mm}
\centering
\begin{minipage}{0.3\textwidth}
\centering
\includegraphics[scale = 0.33]{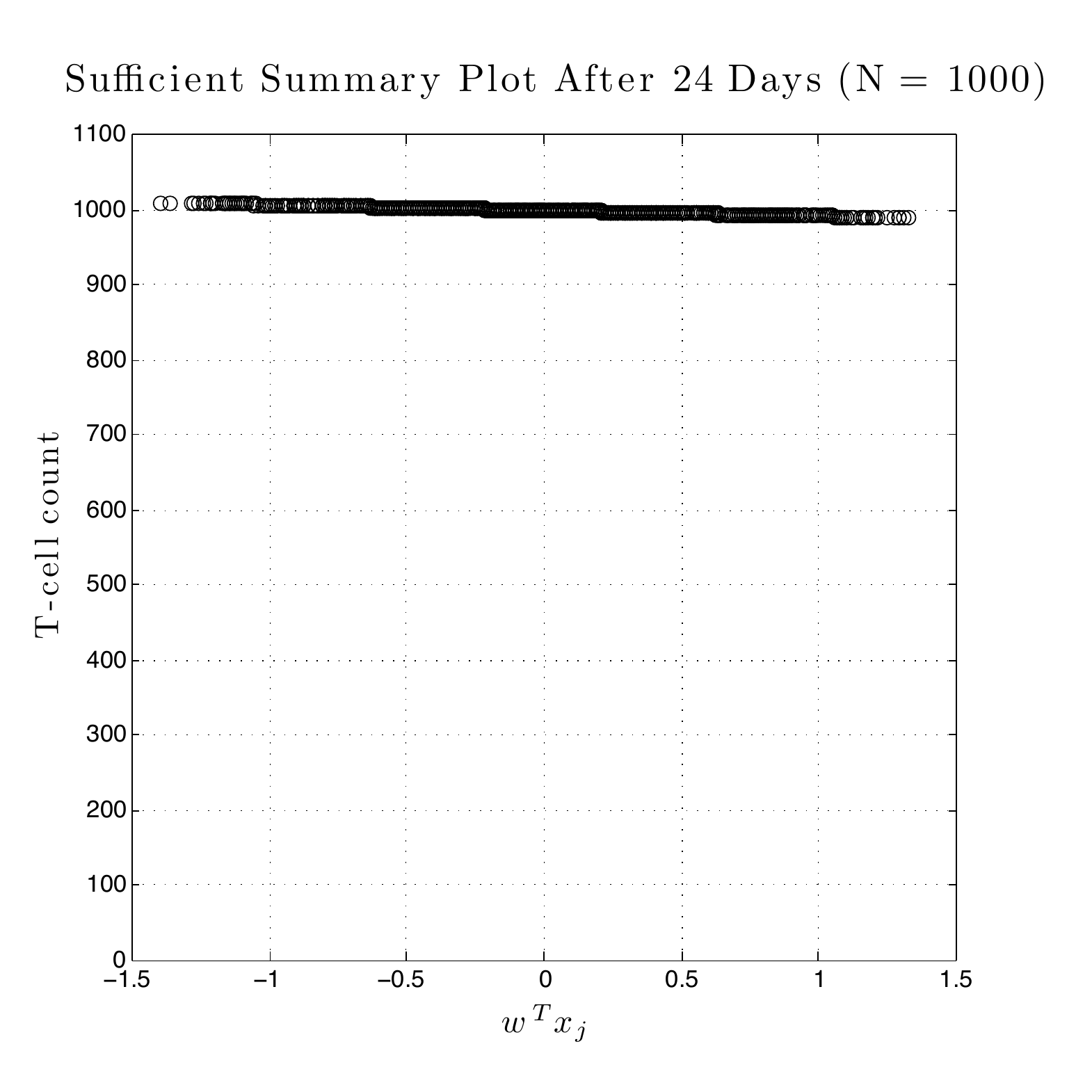}
\end{minipage}
\hspace{10mm}
%
\begin{minipage}{0.3 \textwidth}
\centering
\includegraphics[scale = 0.33]{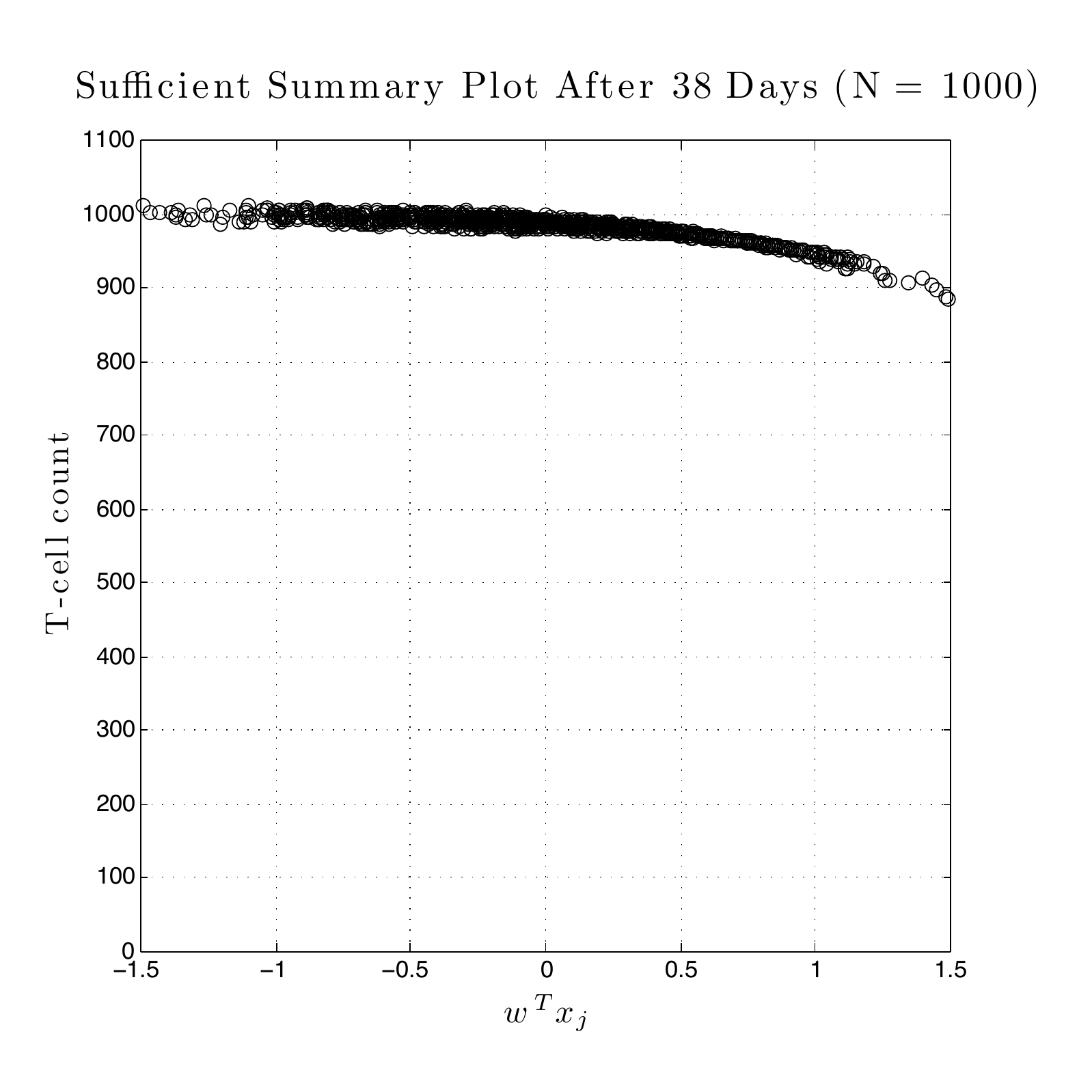}
\end{minipage}
\hspace{10mm}
\begin{minipage}{0.3 \textwidth}
\centering
\includegraphics[scale = 0.33]{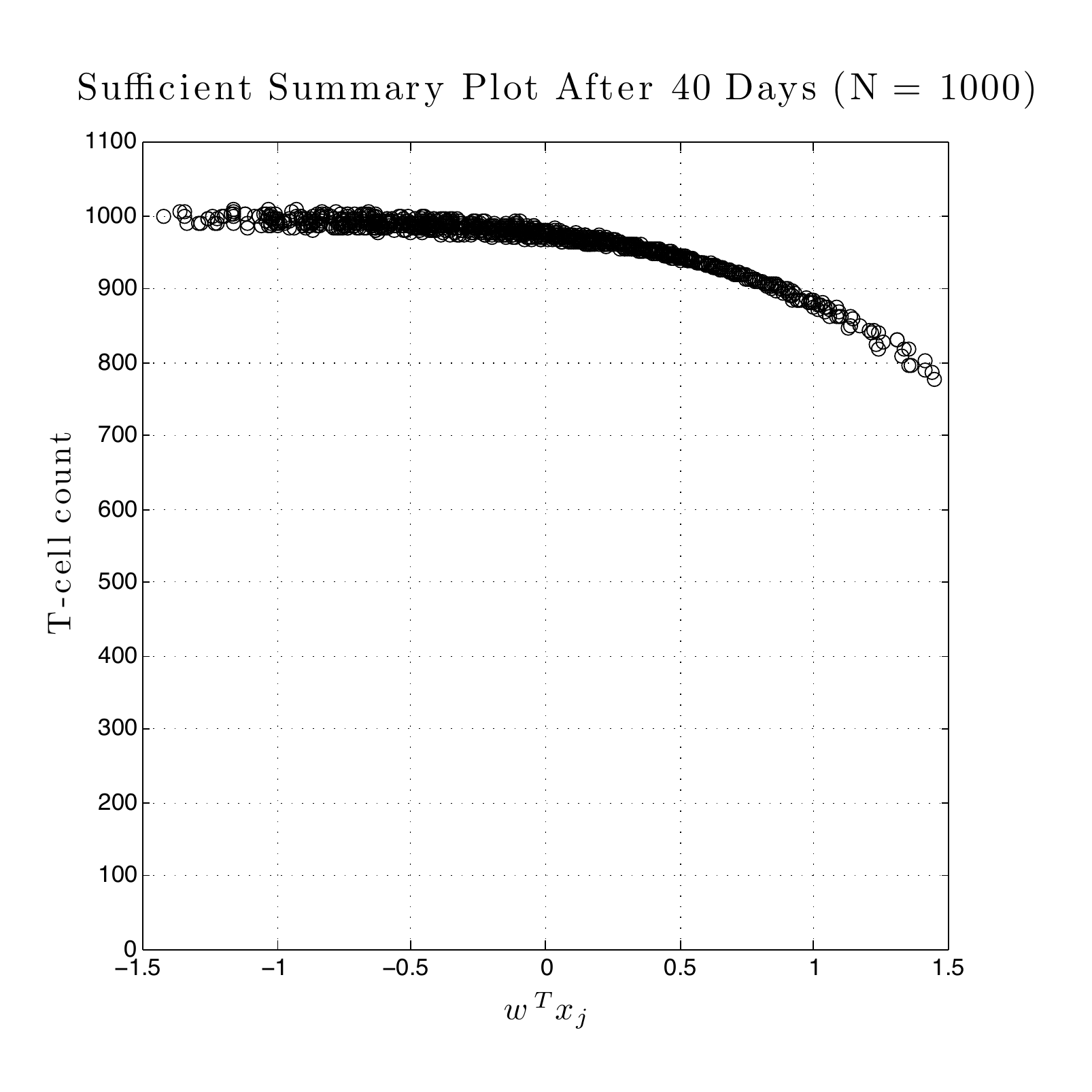}
\end{minipage}
\\
\vspace{-0.15in}
\hspace{-13mm}
\begin{minipage}{0.3 \textwidth}
\centering
\includegraphics[scale = 0.33]{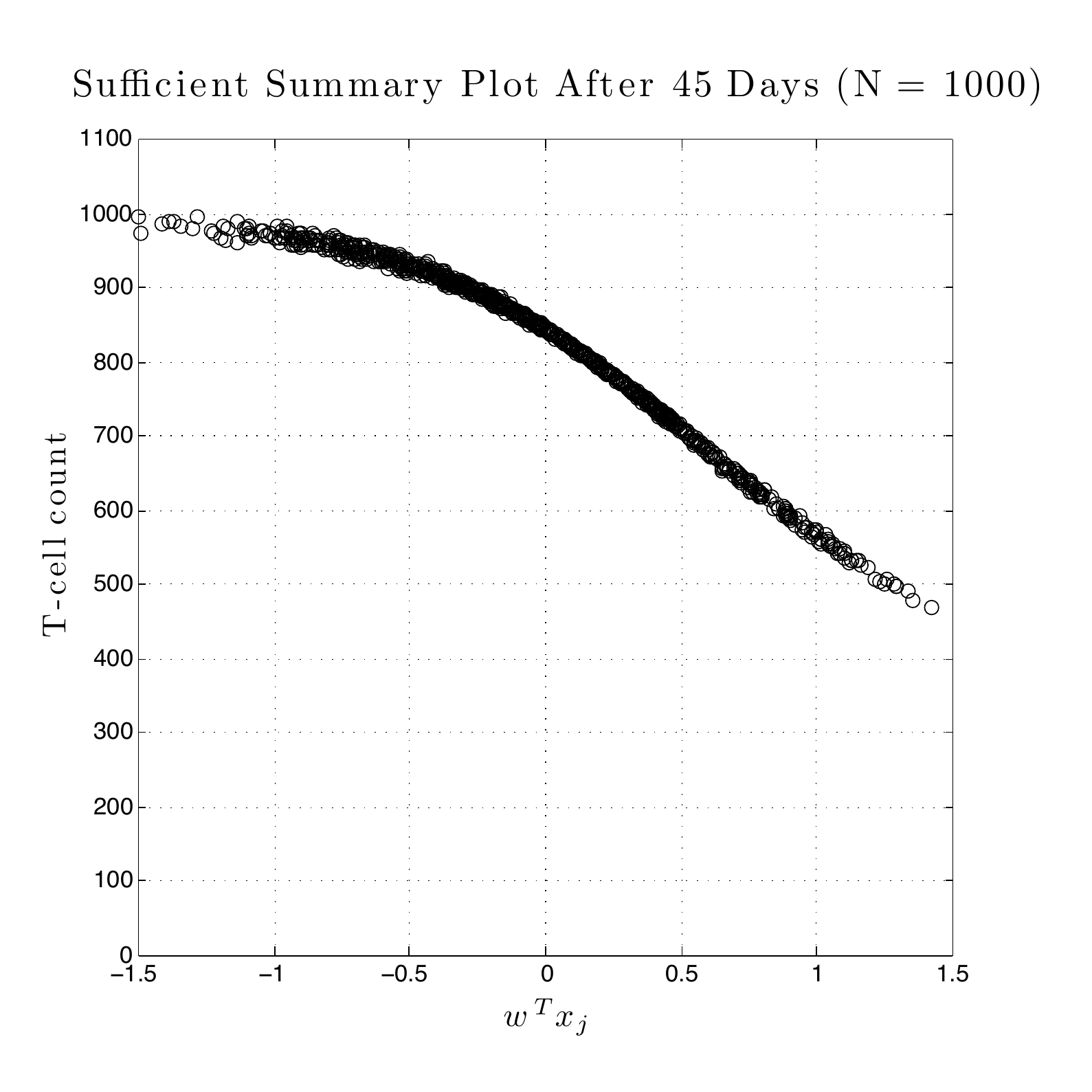}
\end{minipage}
\hspace{10mm}
\begin{minipage}{0.3 \textwidth}
\centering
\includegraphics[scale = 0.33]{SSP15.pdf}
\end{minipage}
\hspace{10mm}
%
\begin{minipage}{0.3 \textwidth}
\centering
\includegraphics[scale = 0.33]{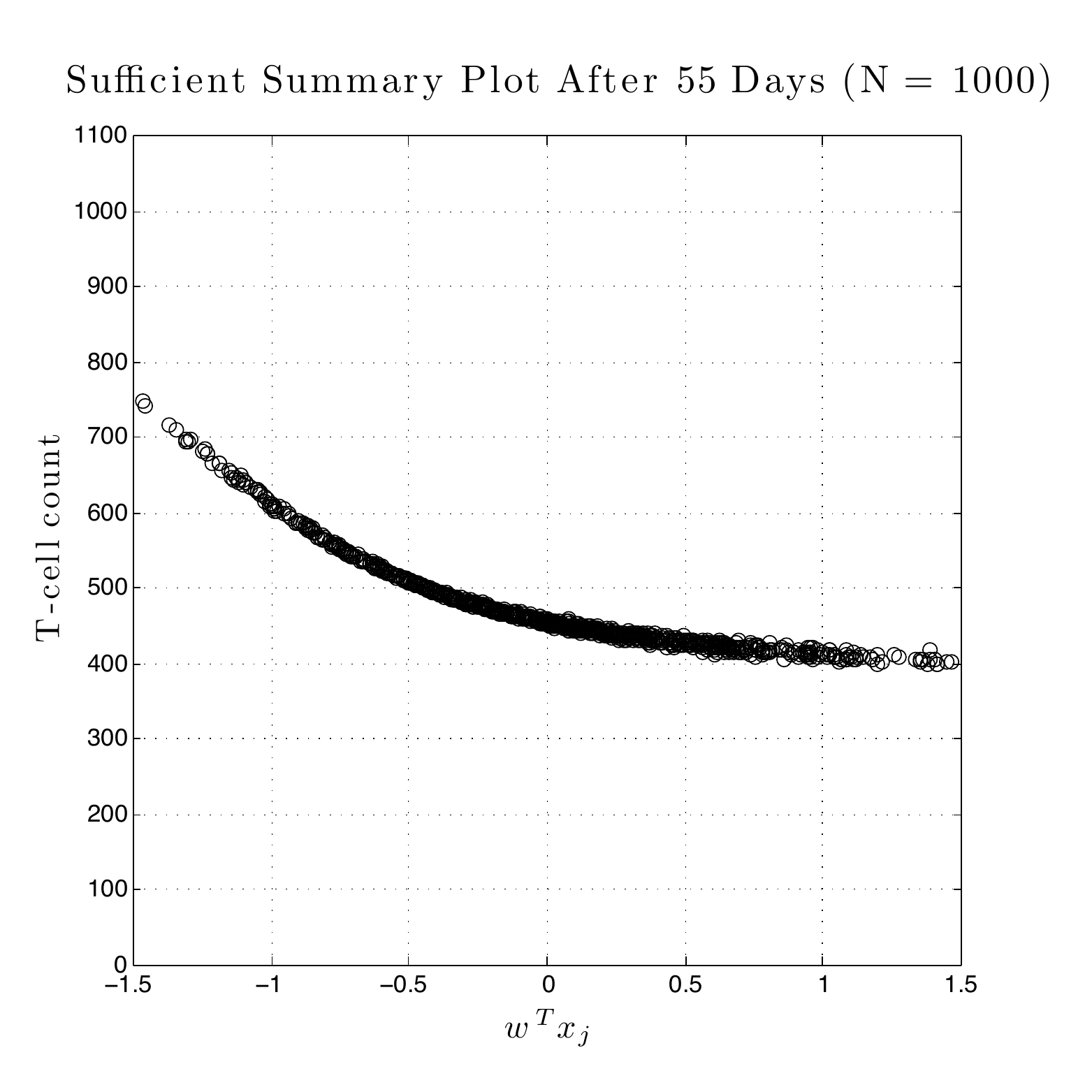}
\end{minipage}
\\
\vspace{-0.2in}
\captionof{figure}{Sufficient summary plots throughout the course of infection - Acute stage.}
\label{fig:SSPs1}
\end{minipage}

\noindent\begin{minipage}{\textwidth}
\vspace{-0.1in}
\hspace{-13mm}
\begin{minipage}{0.3 \textwidth}
\centering
\includegraphics[scale = 0.33]{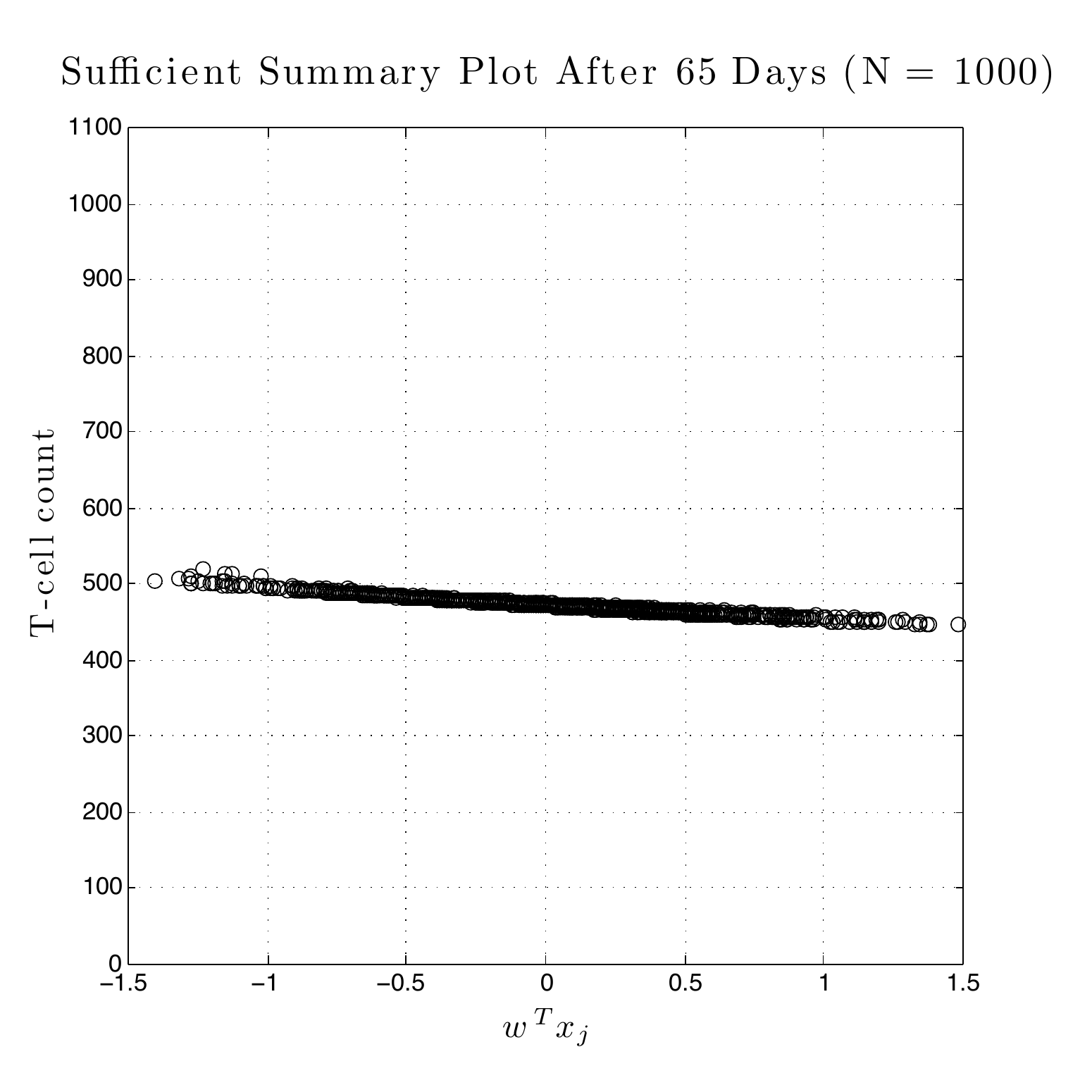}
\end{minipage}
\hspace{8mm}
\begin{minipage}{0.3 \textwidth}
\centering
\includegraphics[scale = 0.33]{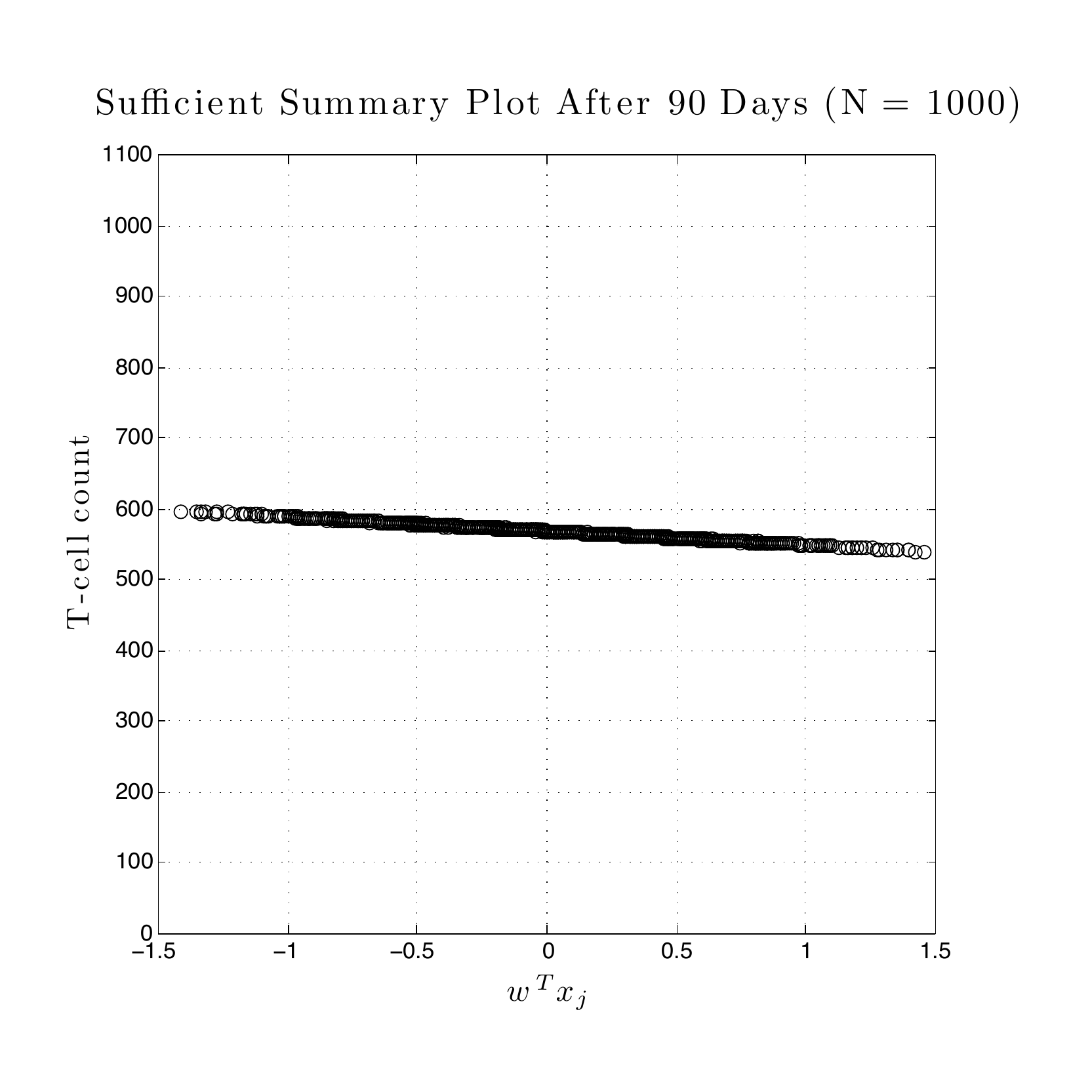}
\end{minipage}
\hspace{11mm}
\begin{minipage}{0.3 \textwidth}
\centering
\includegraphics[scale = 0.33]{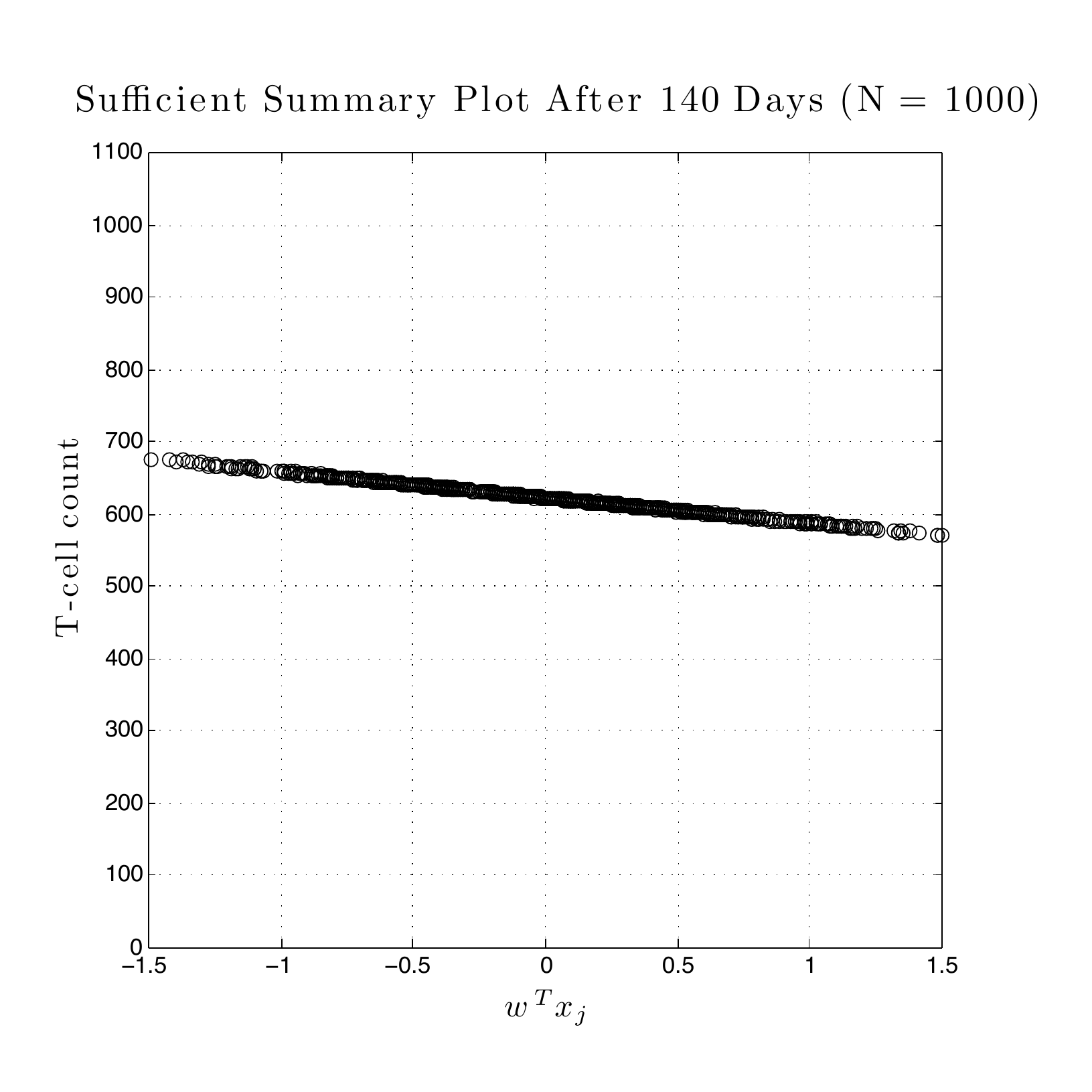}
\end{minipage}

\vspace{-0.25in}
%
\hspace{-15mm}
\begin{minipage}{0.3 \textwidth}
\centering
\includegraphics[scale = 0.33]{SSP33.pdf}
\end{minipage}
\hspace{10mm}
\begin{minipage}{0.3 \textwidth}
\centering
\includegraphics[scale = 0.33]{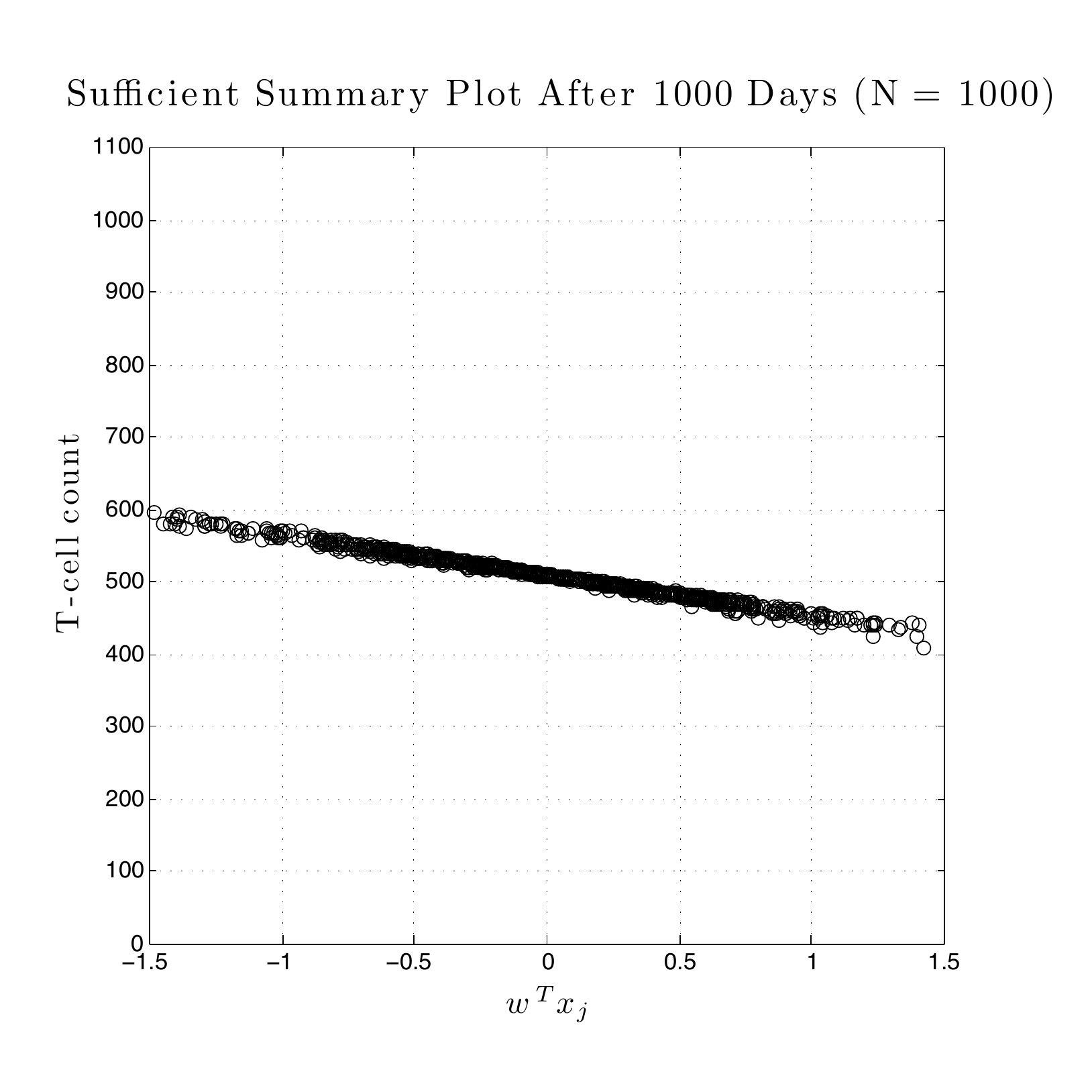}
\end{minipage}
\hspace{10mm}
%
\begin{minipage}{0.3 \textwidth}
\includegraphics[scale = 0.33]{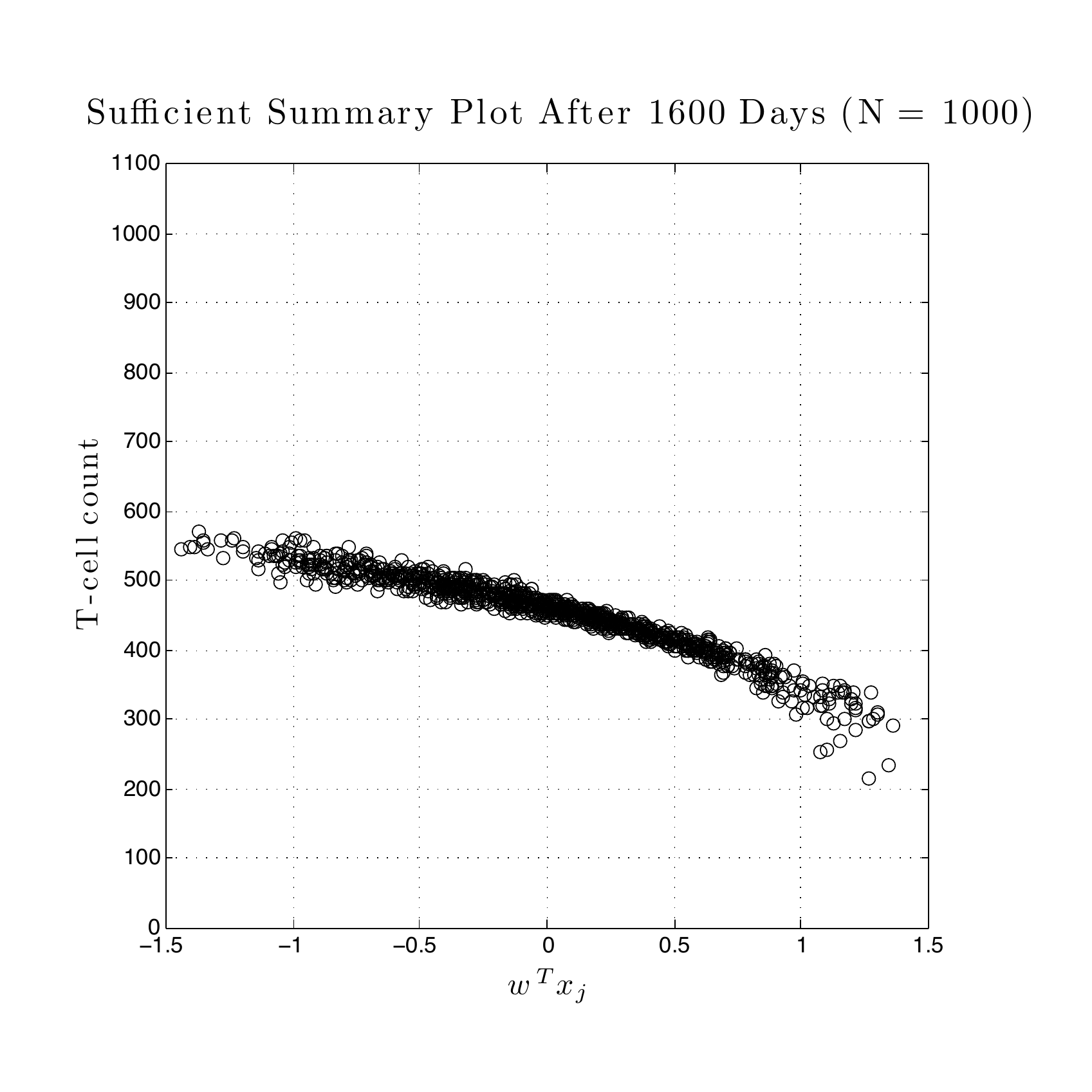}
\end{minipage}
\vspace{-0.2in}
\captionof{figure}{Sufficient summary plots throughout the course of infection - Chronic stage.}
\label{fig:SSPs2}
\end{minipage}

\noindent\begin{minipage}{\textwidth}
\vspace{-0.1in}
\hspace{-15mm}
\begin{minipage}{0.3 \textwidth}
\centering
\includegraphics[scale = 0.33]{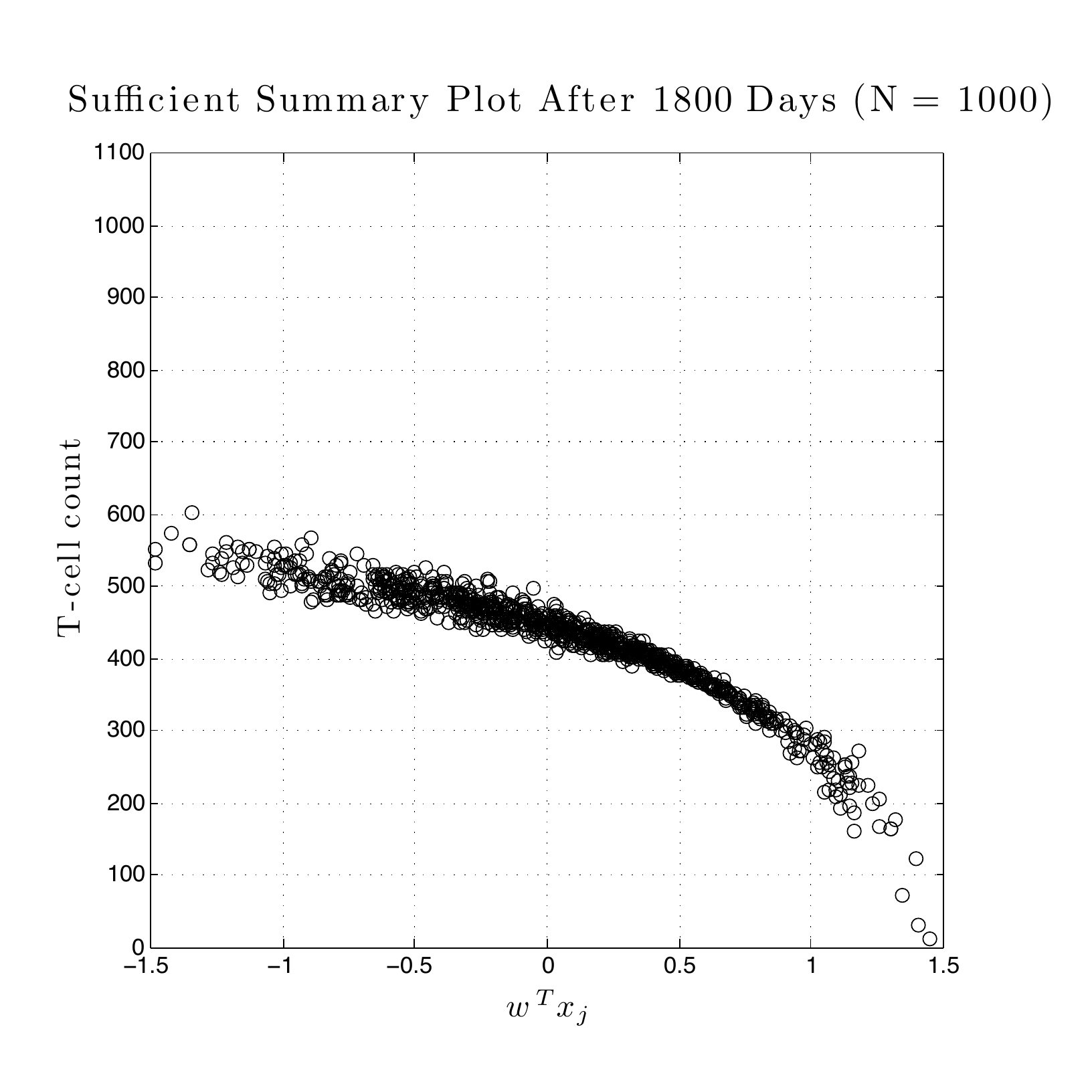}
\end{minipage}
\hspace{8mm}
\begin{minipage}{0.3 \textwidth}
\centering
\includegraphics[scale = 0.33]{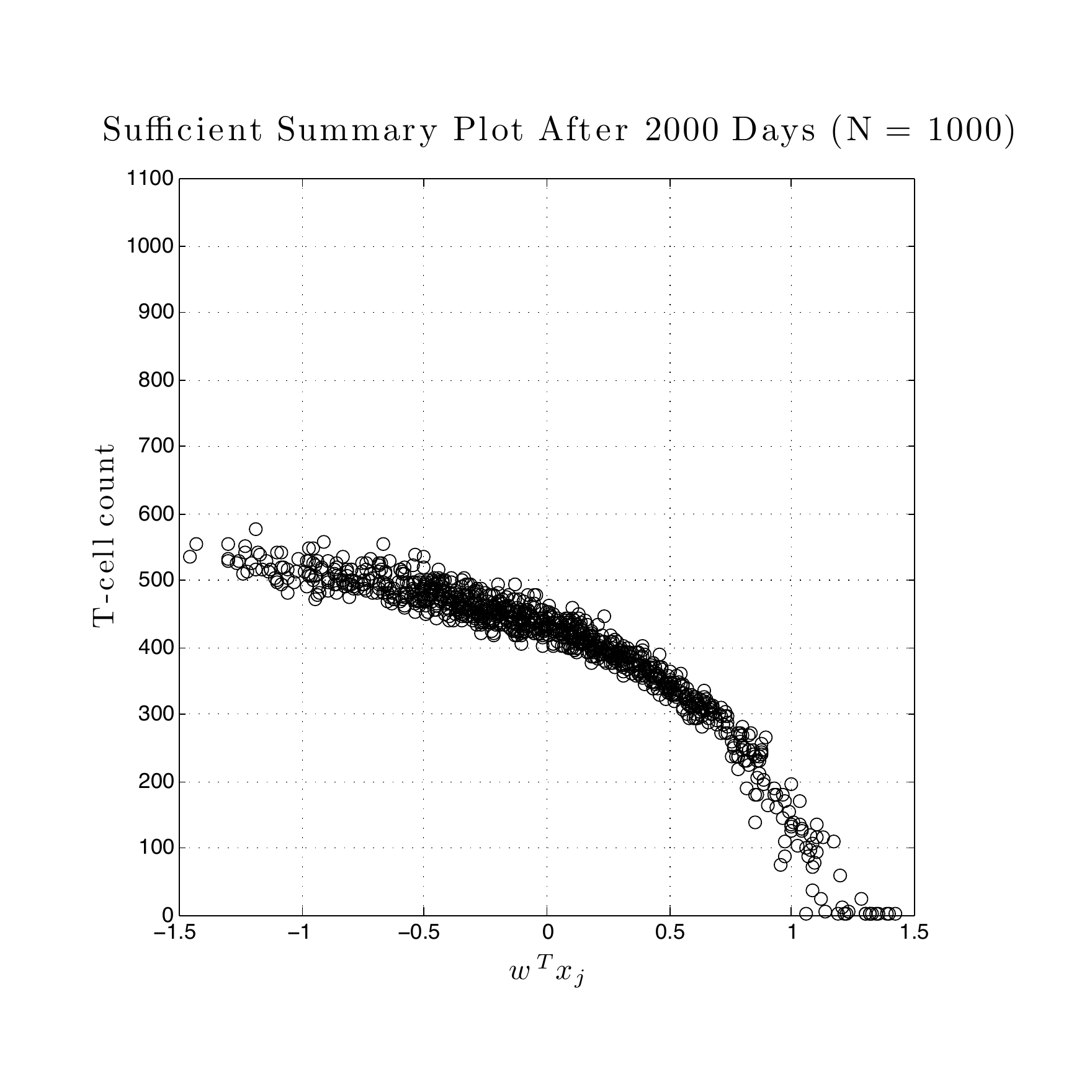}
\end{minipage}
\hspace{12mm}
\begin{minipage}{0.3 \textwidth}
\centering
\includegraphics[scale = 0.33]{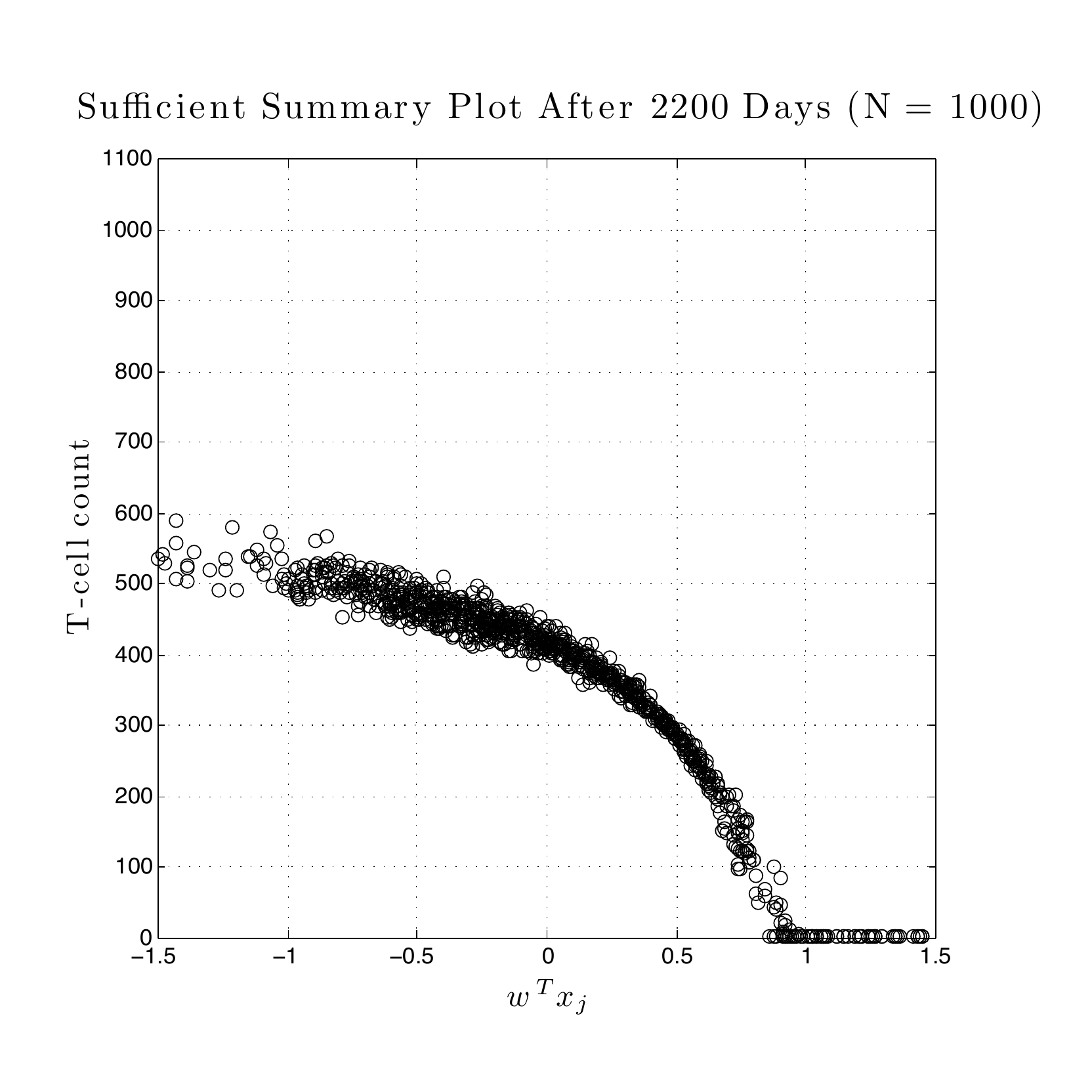}
\end{minipage}

\vspace{-0.3in}
\hspace{-15mm}
\begin{minipage}{0.3 \textwidth}
\centering
\includegraphics[scale = 0.33]{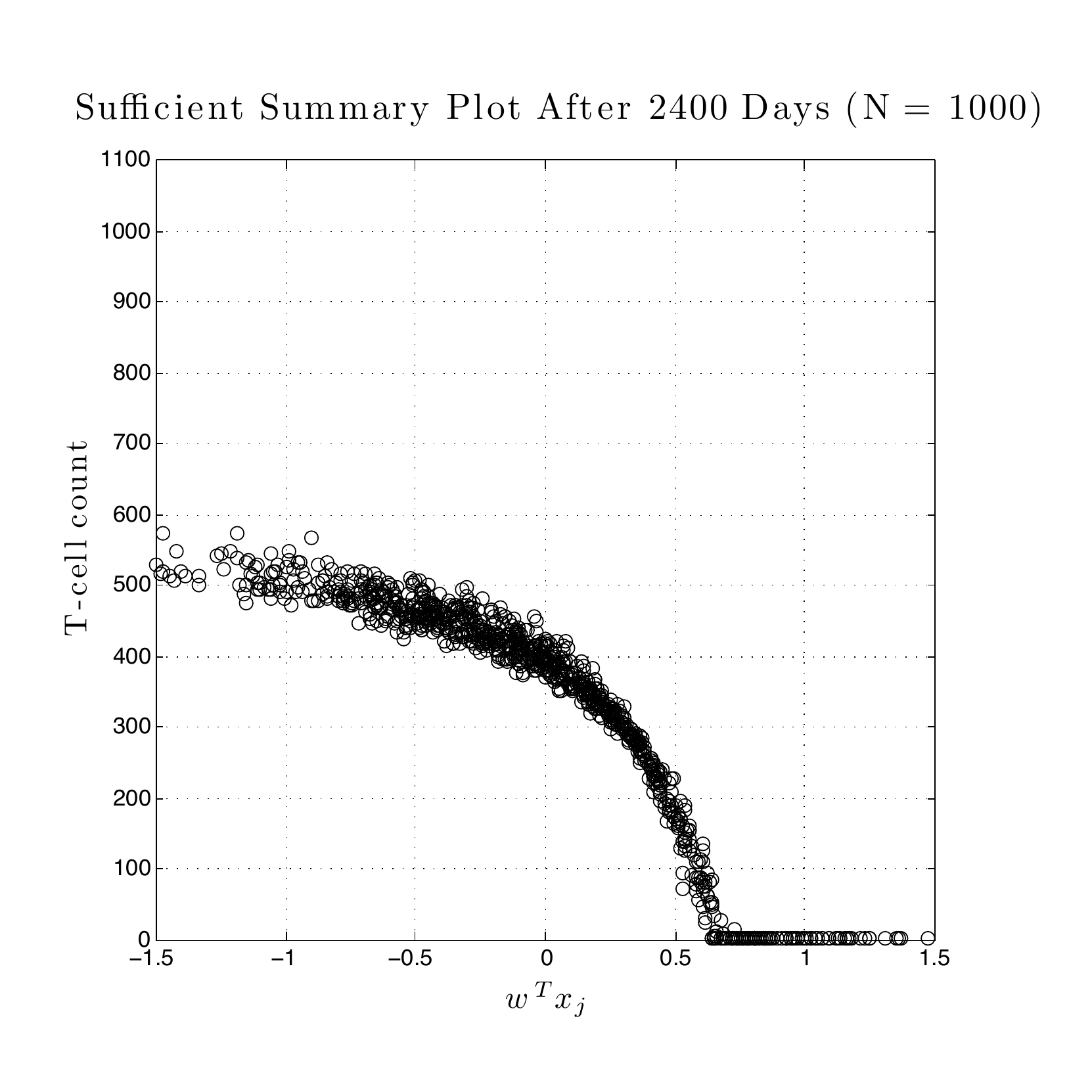}
\end{minipage}
%
\hspace{10mm}
\begin{minipage}{0.3 \textwidth}
\centering
\includegraphics[scale = 0.33]{SSP52.pdf}
\end{minipage}
\hspace{10mm}
\begin{minipage}{0.3 \textwidth}
\centering
\includegraphics[scale = 0.33]{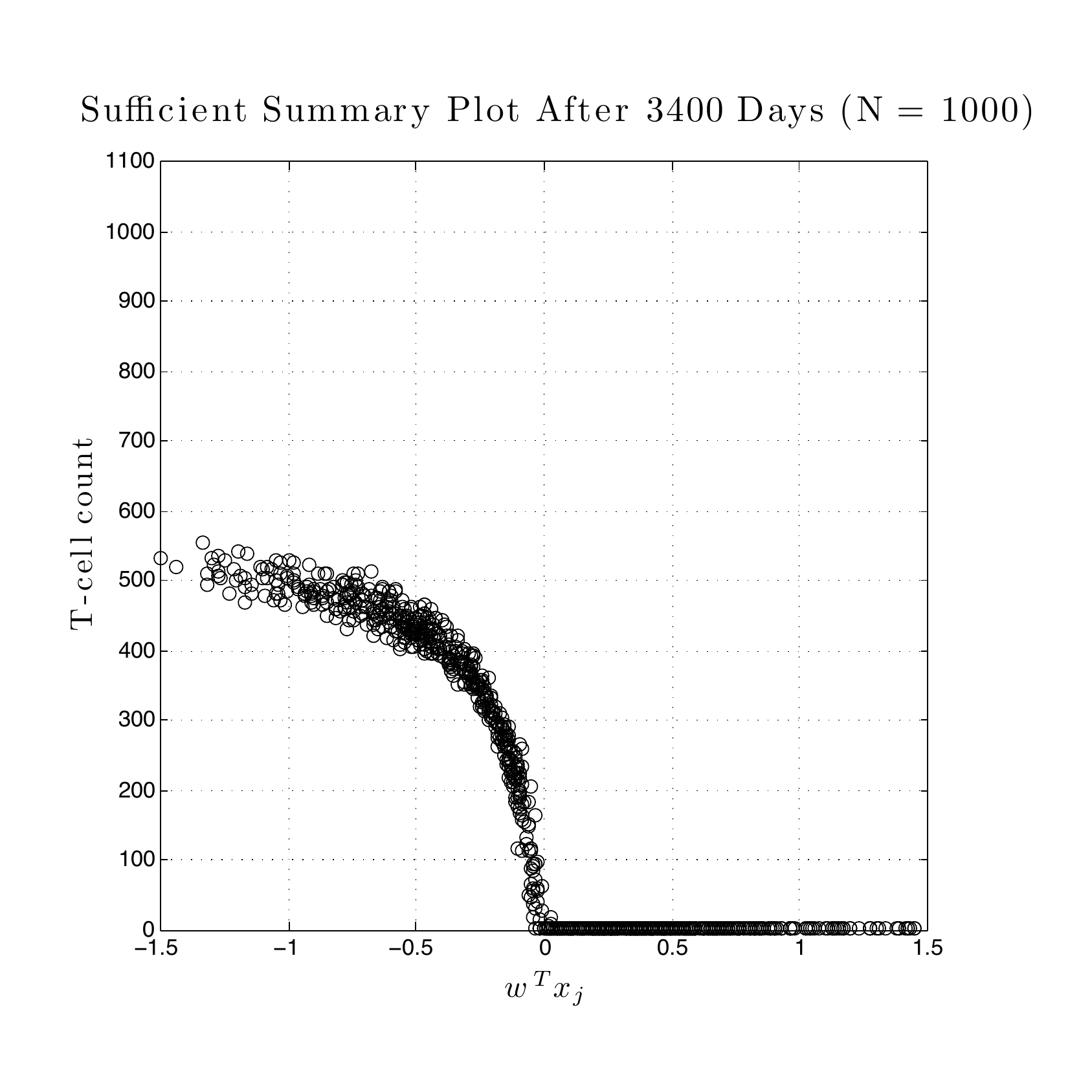}
\end{minipage}
\vspace{-0.2in}
\captionof{figure}{Sufficient summary plots throughout the course of the infection - progression to AIDS}
\label{fig:SSPs3}
\end{minipage}
\vspace{0.1in}
 
\end{document}